\documentclass[twocolumn]{article}

\usepackage{microtype}
\usepackage{hyperref}
\usepackage{float}
\usepackage{authblk}

\usepackage[makeroom]{cancel}
\usepackage{cite}                          
\usepackage{amsmath}
\usepackage{amsthm}
\usepackage{amssymb}
\usepackage{footnote}
\usepackage{array}

\usepackage{calc}
\usepackage{tikz}
\usetikzlibrary{intersections, arrows, shapes, fit, calc, positioning}
\usetikzlibrary{automata}
\usetikzlibrary{fadings,through}
\usetikzlibrary{decorations.markings}

\tikzstyle{vecArrow} = [thick, decoration={markings,mark=at position
    1 with {\arrow[semithick]{open triangle 60}}},
    double distance=1.4pt, shorten >= 5.5pt,
    preaction = {decorate},
    postaction = {draw,line width=1.4pt, white,shorten >= 4.5pt}]

\definecolor{darkgreen}{rgb}{0, 0.5, 0}
\definecolor{lightgreen}{rgb}{0.9,1,0.9}
\definecolor{lightred}{rgb}{1,0.7,0.7}
\definecolor{darkblue}{rgb}{0.3,0.3,1}
\definecolor{lightyellow}{rgb}{1,1,0.5}
\definecolor{darkyellow}{rgb}{0.8,0.8,0}
\definecolor{darkred}{rgb}{0.72,0.04,0.04}
\definecolor{BlueList}{HTML}{0078B4}

\def\showmore{\color{darkgreen} \bf}

\usepackage{todonotes}       

\usepackage{xspace}
\usepackage[inline]{enumitem}
\usepackage{listings}

\usepackage{graphicx}
\usepackage{multirow}
\usepackage{array}
\usepackage[strings]{underscore}
\usepackage{epic}
\def\BibTeX{{ \rm B\kern-.05em{\sc i\kern-.025em b}\kern-.08em
    T\kern-.1667em\lower.7ex\hbox{E}\kern-.125emX}}

\usepackage[linesnumbered,lined,ruled]{algorithm2e}
\newcolumntype{?}{!{\vrule width 1pt}}

\usepackage[capitalize, nameinlink]{cleveref}
\crefname{section}{Sec.}{Sec.}
\crefname{algorithm}{Alg.}{Alg.}
\crefname{figure}{Fig.}{Fig.}
\crefname{proposition}{Prop.}{Prop.}
\crefname{table}{Table}{Tables}
\crefname{definition}{Def.}{Def.}
\crefname{theorem}{Thm.}{Thm.}
\crefname{apptable}{Appendix Table}{Appendix Table}

\lstdefinestyle{MyCustomC}{language=C,tabsize=2}
\def\obfor{{\scshape For}\xspace}
\def\obsplit{{\scshape Split}\xspace}
\def\binsec{{\scshape Binsec}\xspace}
\def\triton{{\scshape Triton}\xspace}
\def\klee{{KLEE}\xspace}
\def\obbanescu{{\scshape Range Divider}\xspace}
\def\obwrite{{\scshape Write}\xspace}

\def\tasseup{\vspace*{-.3cm}}


\newtheoremstyle{test}%
  {\topsep}
  {\topsep}
  {\itshape}
  {0pt}
  {\indent\bfseries}
  { }
  { }
  {\thmname{#1}\thmnumber{ #2}\thmnote{ (#3)}}

\theoremstyle{test}
\newtheorem{theo}{Theorem}
\newtheorem{defi}{Definition}

\newcommand{\obs}[1]{#1'}

\makeatletter

\def\pg#1{\noindent\textbf{#1.}}
\renewcommand{\paragraph}[1]{\smallskip\pg{#1}}

\newcolumntype{R}[1]{>{\raggedleft\let\newline\\\arraybackslash\hspace{0pt}}m{#1}}
\newcolumntype{C}[1]{>{\centering\let\newline\\\arraybackslash\hspace{0pt}}m{#1}}

\definecolor{lightgrey}{rgb}{0.9,0.9,0.9}
\definecolor{darkgrey}{rgb}{0.4,0.4,0.4}
\definecolor{darkred}{rgb}{0.72,0.04,0.04}

\newcommand{\chain}[1]{%
\begin{tikzpicture}[#1]%
    \draw (0,1ex) -- (1ex,1ex) ;
    \draw (1ex,1ex) -- (1ex,0) ;
    \draw[->] (1ex,0) -- (2ex,0) ;
\end{tikzpicture}%
}

    \def\baddef/{\textcolor{darkred}{\bf $\times$}}
    \def\verybaddef/{\textcolor{darkred}{\bf $\times\times$}}
    \def\gooddef/{\textcolor{darkgreen}{{\bf \checkmark}}}

\begin{document}

\def\mytitle{How to Kill Symbolic Deobfuscation for Free}
\title{\mytitle\\
{\smaller (or: Unleashing the Potential of Path-Oriented Protections)}\\
}

\author[1]{Mathilde Ollivier}
\author[1]{Sebastien Bardin}
\author[1]{Richard Bonichon}
\author[2]{Jean-Yves Marion}
\affil[2]{LORIA - \it Nancy, France}
\affil[1]{CEA LIST - \it Paris-Saclay, France}
\date{}

\maketitle

\smallskip

\begin{abstract}
    Code obfuscation is a major tool for protecting software intellectual property from attacks such as 
    reverse engineering or code tampering. 
    Yet, recently proposed (automated) attacks based on 
    Dynamic Symbolic Execution (DSE) shows very promising results, 
    hence threatening software integrity. 
    Current defenses are not fully satisfactory, being either 
    not efficient against symbolic reasoning, or affecting runtime performance
    too much, 
    or being too easy to spot. 
    We present and study a new class of anti-DSE protections coined as path-oriented protections
    targeting the weakest spot of DSE, namely path exploration.
    We propose a lightweight, efficient, resistant and analytically proved class of obfuscation 
    algorithms designed to hinder DSE-based attacks. Extensive evaluation  
    demonstrates that these approaches critically counter symbolic
    deobfuscation while yielding only a very slight overhead. 
\end{abstract}

\section{Introduction}

\paragraph{Context}
 Reverse engineering and code tampering are widely used to extract proprietary
assets (e.g., algorithms or cryptographic keys) or bypass security checks from software. 
{\it Code protection} techniques precisely seek to prevent, or at least make difficult, such {\it man-at-the-end} attacks, 
where the attacker has total control of the environment running the software under attack.  
Obfuscation \cite{Collberg97ataxonomy,Collberg:2009:SSO:1594894} aims at hiding a program's behavior by transforming 
its executable code in such a way that the behavior is conserved  but the program becomes much harder to understand. 

Even though obfuscation techniques are quite resilient against basic automatic reverse
engineering (including static attacks, e.g.~disassembly,  and  dynamic attacks, e.g.~monitoring), code analysis improves quickly \cite{Schrittwieser:2016}. 
Recent attacks based on
\textit{Dynamic Symbolic Execution} (DSE, a.k.a.~\textit{concolic execution})
\cite{CadarS13,GodefroidLM12,SchwartzAB10} use logical formulas to represent input
constraints along an execution path, and then automatically solve these
constraints to discover new execution paths. DSE appears to be very efficient  
against existing obfuscations
\cite{BanescuCGNP16,DBLP:conf/sp/BardinDM17,SalwanBarPot18,YadegariJWD15,CooganLD11},
combining the best of dynamic and semantic analysis.

\paragraph{Problem} \textit{The current state of symbolic deobfuscation is actually pretty unclear}.  
Dedicated protections have been proposed, mainly based on hard-to-solve
predicates, like Mixed Boolean Arithmetic formulas (MBA)
\cite{ZhouMGJ07} or cryptographic hash functions \cite{SharifLGL08}. Yet the
effect of complexified constraints on automatic solvers is hard to predict
\cite{BanescuCP17}, while cryptographic hash functions are easy to spot, may
induce significant overhead and are amenable to key extraction
attacks (possibly by DSE).

 On the other hand, DSE has been fruitfully applied on malware and legit codes
 protected by state-of-the-art tools and methods, including virtualization,
 self-modification, hashing or MBA
 \cite{YadegariJWD15,SalwanBarPot18,DBLP:conf/sp/BardinDM17}.
 The recent systematic experimental evaluation of symbolic deobfuscation by
 Banescu et al. \cite{BanescuCGNP16} shows that most standard obfuscation
 techniques do not seriously impact DSE. 
 Only nested virtualization seems to provide a good
 protection, assuming the defender is ready to pay a high cost in terms of
 runtime and code size \cite{SalwanBarPot18}.  

\paragraph{Goals and Challenges} We want to propose a new class of dedicated anti-DSE 
obfuscation techniques to render automated attacks based on
symbolic execution inefficient. These techniques should be {\it strong} -- making 
DSE intractable in practice,
and {\it lightweight} -- with very low overhead in both code size and runtime performance. 
While most anti-DSE defenses try to break the 
symbolic reasoning part of DSE (constraint solver), we instead target 
its real weak spot, namely path exploration. Banescu et al.~\cite{BanescuCGNP16} 
present one such specific obfuscation scheme but with a large space overhead and no 
experimental evaluation. We aim at  proposing a general framework to understand 
such obfuscations and to define new  schemes  {\it both strong and lightweight}. 

\paragraph{Contribution} We study \emph{path-oriented} protections, a class of
protections seeking to hinder DSE by substantially increasing the number of
feasible paths within a program. 

\begin{itemize}[wide]
    \item We detail a formal framework describing \emph{path-oriented} protections
        (\cref{sec:approach}). 
        We characterize their desirable properties  --- namely {\it tractability},
        {\it strength}, and the key criterion of  {\it single value path} (SVP).  
The framework is  {\it predictive}, in the sense that our classification is
confirmed by experimental evaluation (\cref{sec:xp-eval}), allowing both to 
shed new light on the few existing path-oriented protections and to provide 
guidelines to design  new ones. 
        In particular, no existing protection~\cite{BanescuCGNP16}   achieves both tractability
        and optimal strength (SVP).  As a remedy, we propose {\it the  first two obfuscation schemes}  
         achieving both {\it tractability and optimal strength} (\cref{subsec:new-patterns}). 

    \item We highlight the importance of the {\it anchorage policy}, 
        i.e.~the way to choose where to insert protection in the code, in terms of 
          protection efficiency and robustness. Especially, we identify a way to achieve 
          {\it optimal composition} of path-oriented protections  (\cref{sec:optimal-composition}),  
          and  to completely prevent taint-based and slic-based attacks (two powerful code-level 
          attacks against obfuscation), coined as {\it resistance by design} (\cref{subsec:resbydesign}).  

     \item  We conduct extensive experiments (\cref{sec:secret-finding}) with two different attack scenarios --- 
        \textit{exhaustive path coverage} (\cref{sec:path-exploration}) and 
        \textit{secret finding}. 
        Results confirm that path-oriented protections are much stronger against DSE attacks than standard  protections (including nested virtualization) for only a slight overhead.  
        Moreover, while existing techniques \cite{BanescuCGNP16} can still be weak in some scenarios (e.g., secret finding),   
        our {\it new optimal schemes cripple 
        symbolic deobfuscation at essentially no cost in any setting}. 
Finally, experiments against slice, pattern-matching and taint attacks confirm the quality of our robust-by-design mechanism.    

\end{itemize}

\noindent As a practical outcome, we propose a  new {\it hardened deobfuscation
  benchmark} (\cref{sec:application:hardened}),  
currently out-of-reach of symbolic engines, in order to 
extend existing obfuscation  benchmarks
\cite{tigresschallenge,BanescuCGNP16,SalwanBarPot18}.     

\paragraph{Discussion}  
We study a powerful class of protections against symbolic deobfuscation,
based on a careful analysis of DSE --   we target its weakest point (path
exploration) when other dedicated methods usually aim at its strongest point
(constraint solving and ever-evolving SMT solvers). We propose a predictive framework allowing 
to understand these protections, as well as several concrete protections 
impacting DSE more than three levels of virtualization at essentially no cost. We expect them  
to  be also  efficient against other semantic 
attacks 
\cite{B09,DBLP:conf/popl/HenzingerJMS02}
  (cf.~\cref{sec:discussion}).  
From a methodological point of view, this work extends recent attempts at
rigorous evaluation of obfuscation methods. We 
provide both an analytical evaluation, as Bruni et al. \cite{BruniGG18} for anti-abstract model checking,  
and a refinement of the experimental setup
initiated by Banescu et al. \cite{BanescuCGNP16}.

\section{Motivation}

\subsection{Attacker model}\label{sec:attacker-model}

\paragraph{Goal} 
We consider man-at-the-end scenarios where the attacker has full access to a potentially protected code under attack.
The attacker has just the binary code and no access to the source code.
The attacker model and the methodology follows closely the survey by Schrittwieser et al.~\cite{Schrittwieser:2016}. 
In order to be more concrete, we will focus on the following (intermediate) goals:
\begin{enumerate*}
    \item \label{it:epe} \textit{Exhaustive Path Exploration}. Covering every feasible path
        of the binary allows the attacker to retrieve a consistent Control Flow Graph and
        understand what the original program performs.
   \item \label{it:secret} \textit{Secret Finding}. Focusing on a specific part
        of the code (\textit{e.g.} license checks) and try to understand or retrieve a secret (\textit{e.g.} a key).
\end{enumerate*}

\paragraph{Capacity} we assume an {\it all-powerful symbolic adversary}, that is, this adversary can run a
correct and {\it complete} Dynamic Symbolic Execution (DSE). In practice, symbolic
engines are correct --- every path discovered is actually feasible --- but
incomplete --- they can be tricked into missing feasible paths
\cite{YadegariD15} or the underlying solver may timeout. 
This adversary can  also perform additional code-analysis based attack steps, such as {\it slicing}-
\cite{SrinivasanR16}, pattern-matching  or {\it tainting}-based \cite{SchwartzAB10} code simplifications.

\paragraph{Caveat} In the remains, we always consider this attacker model. There is one caveat that it is worth mentioning now: 
\textit{ part of our experimental evaluations are done from  source codes, i.e.
in our experiments the attacker has sometimes the source code. 
The reason is that state-of-the-art source-level DSE tools are much more efficient than binary-level ones,  
and that there is no good state-of-the-art tools to perform slice or taint attacks on binary codes.  
\textbf{Our experimental conditions are  much more in favor of the attacker}, 
and as a result they show that our approach is all the more effective.  
}

\begin{figure}[!htbp]
  \centering
        \begin{lstlisting}[backgroundcolor=\color{lightgrey}]
int check_char_0(char chr){ 
  char ch = chr;
  ch ^= 97;
  return (ch == 31);
}

/* ... 9 other checks ...  */

int check_char_10(char chr){ /* ... */ }

int check(char* buf) { 
  int retval = 1;
  retval *= check_char_0(buf[0]);
  /* ... check buf[1] to buf[9] ... */
  retval *= check_char_10(buf[10]);
  return retval;
}

int main(int argc, char** argv) {
  char* buf = argv[1];
  if (check(buf)) puts("win");
  else puts("lose");
}
\end{lstlisting}
  \tasseup
  \caption{Manticore crackme code structure}  \label{fig:cm-struct}
  \end{figure}

\subsection{Motivating example}
\label{sec:motivating-example}

Let us illustrate  anti-symbolic path-oriented protections on a toy crackme
program\footnote{\url{https://github.com/trailofbits/manticore}}. \cref{fig:cm-struct}
displays a skeleton of its source code. 
\lstinline{main} calls \lstinline{check} to verify 
each character of the 11 bytes input. It then outputs
\lstinline{"win"} for a correct guess, \lstinline{"lose"} otherwise.
Each subfunction \lstinline{check_char_i}$_{i \in [0, 10]}$ hides a secret
character value behind bitwise transformations, like {\sf xor} or {\sf shift}. 
{\it Such a challenge can be easily solved, completely automatically, by
  symbolic execution tools.
  \klee~\cite{DBLP:conf/osdi/CadarDE08} needs 0.03s  (on C
  code) and \binsec~\cite{DBLP:conf/wcre/DavidBTMFPM16} 0.3s (on
  binary code) to both find a winning input and explore all paths.  }

\paragraph{Standard protections} 
Let us now protect the program with standard obfuscations to measure their impact
on symbolic deobfuscation.  We will rely on Tigress~\cite{CollbergMMN12},
a widely used tool for systematic evaluation of deobfuscation methods
\cite{DBLP:conf/sp/BardinDM17,SalwanBarPot18,BanescuCGNP16}, to apply
 (nested) virtualization, a most effective obfuscation \cite{BanescuCGNP16}.
Yet, \cref{tab:motiv} clearly shows that virtualization does not prevent \klee\ from
finding the winning output, though it can thwart path exploration -- but with a
high runtime overhead (40$\times$).

\paragraph{The case for (new) path-oriented protections}
To defend against
symbolic attackers, we thus need  better anti-DSE 
obfuscation: \emph{path-oriented protections}. 
Such protections aim at exponentially increasing the number of paths 
that a DSE-based deobfuscation tool, like \klee, must explore. Two such
protections are \obsplit and \obfor, illustrated in \cref{fig:obf-eg} on function 
\lstinline{check_char_0} of the  example.

\begin{figure}[!htbp]
  \centering
  \begin{minipage}[t]{0.48\linewidth}
    \centering 
       \large\textbf{\obfor}
            \begin{lstlisting}[backgroundcolor=\color{lightgrey}]
int func(char chr){


  char ch = 0;
  <@\textcolor{darkred}{for (int i=0; i$<$chr; i++)}@>
    <@\textcolor{darkred}{ch++;}@>
  ch ^= 97;
  return (ch == 31);
}
\end{lstlisting}
        \end{minipage}
        \begin{minipage}[t]{0.5\linewidth}
        \centering  
        \large\textbf{\obsplit}
            \begin{lstlisting}[backgroundcolor=\color{lightgrey},breaklines]
int func(char chr,ch1,ch2) {
// new input char ch1 and ch2
  char garb = 0 // junk
  char ch = chr;
  <@\textcolor{darkred}{if (ch1 $>$ 60)}@> 
    <@\textcolor{darkred}{garb++;}@>
  <@\textcolor{darkred}{else}@> 
    <@\textcolor{darkred}{garb$--$;}@>
  <@\textcolor{darkred}{if (ch2 $>$ 20)}@> 
    <@\textcolor{darkred}{garb++;}@>
  <@\textcolor{darkred}{else}@>
    <@\textcolor{darkred}{garb$--$;}@>
  ch ^= 97;          
  return (ch == 31);
}
\end{lstlisting}
      \end{minipage}
      \tasseup
      \caption{Unoptimized obfuscation of \lstinline{check_char_0}}
            \label{fig:obf-eg} 
\end{figure}

{ \it For the sake of simplicity, the protections are implemented in a naive
  form, sensitive to slicing  or compiler
  optimizations. Robustness is discussed afterwards.}  
In a nutshell, \obsplit --- an instance of \obbanescu \cite{BanescuCGNP16} ---
adds a number $k$ of conditional statements depending on new fresh inputs,
increasing the number of paths to explore by a factor of $2^k$.  Also, in {\it this implementation}   
we use a junk variable
\lstinline{garb} and two additional inputs \lstinline{ch1} and \lstinline{ch2}  unrelated to the original code. 
The {\it novel} obfuscation \obfor (\cref{subsec:new-patterns})   adds $k$ 
loops whose upper bound depends on distinct input
 bytes and which recompute a value that will be used later, expanding the number of paths to explore 
 by a factor of  $2^{8 \cdot k}$ -- assuming a 8-bit \lstinline{char} type. {\it This implementation} 
does not introduce any junk variable nor additional input.  
In both cases, the obfuscated code relies on the input, 
forcing DSE to explore {\it a priori} all paths. 
\cref{tab:motiv} summarizes the performance of \obsplit and  \obfor.
\begin{enumerate*}
\item[]  Both \obsplit and \obfor do not induce any  overhead, 
\item[]  \obsplit is highly efficient (timeout) against coverage but not against secret finding,  while
        \obfor is highly efficient for both. 
\item[]  \obfor($k=2$) performs already better than \obsplit($k=19$) and
  further experiments (\cref{sec:xp-eval}) shows \obfor to be a much more effective path protection than \obsplit. 
\end{enumerate*}

\begin{table}[!htbp]
    \caption{DSE Attack on the Crackme Example (\klee)} 
    \label{tab:motiv}
    \centering
    \resizebox{\columnwidth}{!}{%
        \begin{tabular}{|c|c|l||C{1.3cm}|C{1cm}||R{1.2cm}|}
        \cline{2-6}
        \multicolumn{1}{c|}{}
        & \multicolumn{2}{c||}{\multirow{3}{*}{\textbf{Obfuscation type}}}
        & \multicolumn{2}{C{2.3cm}||}{\textbf{Slowdown} \newline \itshape Symbolic Execution}
        & \textbf{Overhead} \\
       \cline{4-5}
       \multicolumn{1}{c|}{} & \multicolumn{2}{c||}{} & Coverage & Secret & runtime\\
       \hline
       \multirow{2}{*}{\textbf{Standard}}
       & \multicolumn{2}{l||}{Virt          } & \verybaddef/ & \verybaddef/ & $\times1.1$ \\
       & \multicolumn{2}{l||}{Virt $\times2$} & \baddef/ & \verybaddef/ & $\times1.3$ \\
       & \multicolumn{2}{l||}{Virt $\times3$} & \gooddef/ & \baddef/ & $\times40$ \\
       \hline\hline
       \multirow{6}{*}{\textbf{Path-Oriented}} & & $k=11$ & \verybaddef/ & \verybaddef/ & $\times1.0$ \\
        & SPLIT & $k=15$ & \gooddef/ & \verybaddef/ & $\times1.0$ \\ 
        & \cite{BanescuCGNP16} & $ k=19$ & \gooddef/ & \verybaddef/ & $\times1.0$ \\ 
        \cline{2-6}
        &  & $k=1$ & \gooddef/ & \baddef/ & $\times1.0$ \\ 
        & FOR  &  $k=2$ & \gooddef/ & \gooddef/ & $\times1.0$ \\ 
        & (new) &  $k=3$ & \gooddef/ &  \gooddef/ & $\times1.0$ \\ 
        \hline
    \end{tabular}}

\smallskip
  
    \verybaddef/  t $\leq$ 1s
    \hfill
    \baddef/ 30s $<$ t $<$ 5min
    \hfill
    \gooddef/ time out ($\ge$ 1h30)

    \smallskip
    Unobfuscated case: \klee succeeds in 0.03s
\end{table}

\smallskip

\noindent\textit{\textbf{Question:} How to distinguish a priori between mildly 
effective and very strong path-oriented protections?} 

\smallskip

Note that \textsf{gcc -Ofast} is able to remove this simple \obsplit, as it is not related to the output {\it (slicing attack)}. The basic \obfor resists such attack, 
but  \textsf{clang -Ofast} is able to remove it  by
an analysis akin to a {\it pattern attack}. However, a sightly modified \obfor\ 
(\cref{fig:enhanced-for}) overcomes such optimizations. 

\begin{figure}[!htbp]
  \centering
\begin{lstlisting}[backgroundcolor=\color{lightgrey}]
int func(char chr) {
  int ch = 0; // prevent char overflows 
  <@\textcolor{darkred}{for (int i=0; i$<$ (int) chr; i++)\{}@>
    <@\textcolor{darkred}{if (i \% 2 == 0) ch += 3;}@>
    <@\textcolor{darkred}{if (i \% 2 != 0) ch$--$;}@>
  <@\textcolor{darkred}{\}}@>
  <@\textcolor{darkred}{if (i \% 2 != 0) ch -= 2;}@> // adjust for odd values
  ch ^= 97;
  return (ch == 31);
}
\end{lstlisting}
  \tasseup
  \caption{Enhanced \obfor\ -- \lstinline{check_char_0}}
  \label{fig:enhanced-for}
\end{figure}

\noindent\textit{\textbf{Question:} How to protect path-oriented protections against code analysis-based attacks (slicing, tainting, patterns)?} 

\smallskip

\noindent {\textbf{The goal of this paper} is to define, analyze and explore in a systematic
  way the potential of path-oriented transformations as anti-DSE
  protections. We  define a  {\it predictive} framework (\cref{sec:approach})   
  and propose several new {\it concrete protections} (\cref{sec:strong-schemes-forking-points}).
  In particular, our framework allows to precisely explain why 
 \obfor is experimentally  better than \obsplit. 
We also discuss how path-oriented protections can be made 
 resistant to several types of attacks (\cref{sec:anchorage-policy,sec:threats}). 
}

\section{Background}
\label{sec:background}

\paragraph{Obfuscation}
\label{sec:obfuscation}
Obfuscation \cite{Collberg97ataxonomy} aims at hiding a
program's behavior or protecting proprietary information such as algorithms or
cryptographic keys by transforming the program to protect $\mathcal{P}$ 
into a program $\obs{\mathcal{P}}$ such that $\obs{\mathcal{P}}$ and $\mathcal{P}$ are semantically
equivalent, $\obs{\mathcal{P}}$ is roughly as efficient as $\mathcal{P}$ and $\obs{\mathcal{P}}$ is {\it harder to understand}. 
While 
 it is still unknown whether applicable theoretical criteria of obfuscation 
exist \cite{BarakGIRSVY01}, practical obfuscation techniques and tools 
do. 

Let us touch briefly on three such important techniques. 

\begin{itemize}
    \item \textit{Mixed Boolean-Arithmetic}~\cite{EyrollesGV16,ZhouMGJ07} transforms an arithmetic
and/or Boolean equation into another using a combination of Boolean and
arithmetic operands with the goal to be more complex to understand and more 
difficult to solve by SMT solvers~\cite{VanegueH12,Barrett2018}. 
    \item \textit{Virtualization} and \textit{Flattening} \cite{Wang2000} 
transform  the control flow into an interpreter loop  dispatching every
instruction. {\em Virtualization} even adds a virtual machine interpreter for 
a custom bytecode program encoding the original program semantic. 
Consequently, the visible control flow of the protected program 
is very far from the original control flow. Virtualization can  be {\it nested}, 
encoding the virtual machine itself into another virtual machine.  
    \item \emph{Self-modifying code} and \emph{Packing} insert instructions that
dynamically modify the flow of executed instructions. These techniques seriously  damage  
static analyses by hiding the real instructions. However, extracting the hidden code
can be done by dynamic approaches \cite{DebrayP10,Kang07}, including DSE \cite{YadegariJWD15}. 

\end{itemize}

\paragraph{Dynamic Symbolic Execution (DSE)}
\label{sec:dse}
Symbolic execution \cite{CadarS13} simulates the execution of a program along its
paths, systematically generating inputs for each new discovered branch
condition. This exploration process consider inputs as \textit{symbolic
  variables} whose value is not fixed. The symbolic execution engine follows a
path and each time a conditional statement involving the input is encountered,
it adds a
constraint to the {\it symbolic value} related to this input. Solving 
the constraints automatically (typically with off-the-shelf SMT solvers
\cite{VanegueH12}) then allows to generate {\it new input
values leading to new paths}, progressively covering all paths of the program
-- up to a user-defined bound. The technique has seen strong renewed interest
in the last decade to become a prominent  bug finding technique  \cite{CadarS13,GodefroidLM12,DBLP:conf/sp/ChaARB12}.

When the symbolic engine cannot perfectly  handle some constructs of the
underlying programming language --- like system calls or self-modification
--- the symbolic reasoning is interleaved with a {\it dynamic analysis} allowing
meaningful (and fruitful) approximations -- \textit{Dynamic Symbolic Execution} 
\cite{GodefroidLM12}.  Typically, (concrete) runtime values are used
to complete missing part of path constraints that are then fed to the solver
through \textit{concretization} \cite{DBLP:conf/issta/DavidBFMPTM16}.  This
feature makes the approach especially robust against complicated constructs  
found in obfuscated binary codes, typically packing or self-modification, 
making DSE a strong candidate for automated  
deobfuscation -- {\it symbolic deobfuscation}:  it is as robust as dynamic 
analysis, with the additional ability to infer {\it trigger-based conditions}.  

\section{A framework for path-oriented protections}
\label{sec:approach}

This section presents a framework to evaluate \emph{path-oriented} obfuscations, i.e.~protections  
aiming  at hindering symbolic deobfuscation by taking advantage of path explosion. 

\begin{figure}[!htbp]
    \centering
        \tikzstyle{block} = [draw, rectangle, minimum height=3em]
        \tikzstyle{input} = [coordinate]
        \tikzstyle{output} = [coordinate]
        \tikzstyle{inter} = [coordinate]

        \resizebox{\columnwidth}{!}{%
        \begin{tikzpicture}[auto,node distance=2cm]
            \node [input, name=input] {};
            \node [inter, name=n1, right of=input] {};
            \node [block,right of=n1] (placement) {Find placement};
            \node [inter,name=n3,right of=placement] {};
            \node [block,right of=n3] (forking) {Add forking point};
            \node [inter,name=n2, right of=forking] {};
            \node [output,name=output,right of=n2] {};

            \draw (input) -- node {\textit{code}} (n1);
            \draw [->] (n1) -- (placement);
            \draw [->] (placement) -- (forking);
            \draw (forking) -- (n2);
            \draw [->] (forking) -- node[near end] {\textit{protected code}} (output);

            \draw [color=gray](2,1) rectangle (10,-1.5);
            \node at (10,-1.3) [left]{\textit{obfuscation scheme}};

            \node [block, above of=placement] (anchor) {Anchorage policy};
            \node [block, above of=forking] (class) {Class(es) of forking points};

            \draw [vecArrow] (anchor) -- (placement);
            \draw [vecArrow] (class) -- (forking);
        \end{tikzpicture}}
        \caption{Path-Oriented Obfuscation Framework}
        \label{fig:framework}
\end{figure}
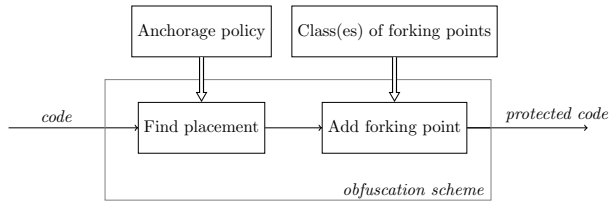

\subsection{Basic definitions}
\label{sec:forking-points}

This paper deals with a specific kind of protections targeting DSE:
\textit{path-oriented} protections. Transforming
a program
$\mathcal{P}$ into $\mathcal{P}'$ using \textit{path-oriented} protections ensures that
$\mathcal{P}'$ is functionally equivalent to $\mathcal{P}$ and aims to
guarantee $\#\Pi' \gg \#\Pi$,  meaning
the number of paths in $\mathcal{P}'$ is much greater that the ones in $\mathcal{P}$. 

The most basic path-oriented protection consists in one \textit{forking point} inserted
in the original code of $\mathcal{P}$.

\begin{defi}[Forking Point]
    A forking point $\mathcal{F}$ is a location in the code  
    that creates at most $\gamma$ new  paths. 
    $\mathcal{F}$ is defined by: 
    an address $a$, a variable $x$ and a capacity $\gamma$.
    It is written $\mathcal{F}(a,x,\gamma)$
\end{defi}

To illustrate this definition, see the snippet of \obsplit  in Figure~\ref{fig:obf-eg}:   
both if-statements define each a forking point of capacity $2$ based on the variable \texttt{ch1} and  \texttt{ch2} respectively.  

\smallskip 

Now,  to obtain a complete path-oriented obfuscation $\mathcal{P}'$ of a program $\mathcal{P}$,
we need to insert $n$ forking points throughout the code of $\mathcal{P}$, hence the notion of 
\textit{obfuscation scheme}  (\cref{fig:framework}). 

\begin{defi}[Obfuscation scheme]
A
\emph{(path-oriented protection) obfuscation scheme} is a function $f(\mathcal{P},n)$ that, for every program
$\mathcal{P}$, inserts $n$ forking points in $\mathcal{P}$. 
It comprises a set of forking points and an {\it anchorage policy}, \textit{i.e.} the
placement method of the forking points.
\end{defi}

\subsection{Desirable obfuscation scheme properties} 
\label{sec:desir-obfusc-scheme}

An ideal obfuscation scheme is both
strong (high cost for the attacker) and cheap (low cost for the defender).  Let
us define these properties more precisely.

\smallskip

The {\bf strength} of an obfuscation scheme is intuitively the expected increase
of the search space for the attacker.
Given an obfuscation scheme $f$, it is defined as
 $\Gamma_f(\mathcal{P},n)= \#\Pi_{f(\mathcal{P},n)}$,
for a program $\mathcal{P}$ and  $n$ forking points to insert. 

\smallskip 

The {\bf cost} is intuitively the {\it maximal} runtime overhead the defender should worry about.  
Given an obfuscation scheme $f$,
 cost is defined by $\Omega_f(\mathcal{P},n)$,  as the {\it maximum} trace size of the obfuscated program $f(\mathcal{P},n)$.
 Formally,  $\Omega_f(\mathcal{P},n)=max_i\{|\tau'_i|\}$ where
 $\{\tau'_i\}$ is the set of execution traces of $f(\mathcal{P},n)$ and $|\tau'_i|$ is the size of the trace.

\smallskip

We seek strong tractable obfuscations, i.e., yielding enough added paths 
 to get a substantial slowdown, with a low runtime overhead. 

\begin{defi}[Strong scheme]
    An obfuscation scheme $f$ is {\em strong} if
    for any program $\mathcal{P}$, we have $\Gamma_f(\mathcal{P},n) \geq 2^{O(n)} \cdot  \#\Pi_{\mathcal{P}}$,
    where $\Pi_{\mathcal{P}}$ is the set of paths of $\mathcal{P}$.
    Putting things quickly, it means that the number of paths to explore is multiplied  by $2^n$.
\end{defi}  

\begin{defi}[Tractable scheme]
    An obfuscation scheme $f$ is {\em tractable} if 
for any program $\mathcal{P}$, $\Omega_f(\mathcal{P},n) \leq max_i\{|\tau_i|\} + O(n)$, 
    where $max_i\{|\tau_i|\}$ is the size of the longest trace of $\mathcal{P}$.
    In other words, it is tractable only if the overhead runtime is linear on $n$.
\end{defi}

\paragraph{Combining schemes}  Scheme composition preserves tractability 
(the definition involves an upper bound) but not necessarily strength
(the definition involves a lower bound). Hence, we need 
{\it optimal composition rules} (\cref{sec:optimal-composition}).  

\subsection{Building stronger schemes}
\label{sec:build-strong-schem}

Strong path-oriented protections, can rely on composition but we saw
in \cref{sec:desir-obfusc-scheme} that it is not straightforward.
But since path-oriented protections first lean on forking execution into
many paths, we should also investigate whether some forking points are better that others.
The best case is to insert $k$ forking points $(\mathcal{F}(a_i,x_i,\gamma_i))_i$
such that it would ensure that each path created by a forking point 
$(\mathcal{F}(a_i,x_i,\gamma_i))_i$ corresponds to only one possible value of the variable $x_i$. 
This leads us to define this type of forking point as a \textit{Single Value Path} (SVP) protection. 

\begin{defi}[Single Value Path]
    \label{def:svp}
    A forking point based on  variable $x$ is \emph{Single Value Path} (SVP)
    if and only if $x$ has only one possible value in each path created by the protection.
\end{defi}

A SVP forking point creates a new path for each possible value of variable
$x$   (\textit{e.g.}, i.e., $2^{32}$ new paths 
are created for an unconstrained C \lstinline{int} variable). 
For example, the \obfor obfuscation shown in Figure~\ref{fig:obf-eg} produces $2^{32}$ paths,
that should be a priori explored since it depends on an input variable.  
SVP forking points is key to ensure that DSE will need to enumerate all possible 
input values of the
program under analysis (thus {\it boiling down to brute forcing}) -- see
~\cref{sec:anchorage-policy}. 

\section{Concrete Path-oriented protections}
\label{sec:strong-schemes-forking-points}

This section   reviews existing path-oriented protection schemes within the 
framework of \cref{sec:approach}, but also details new such schemes achieving both 
tractability and optimal strength (SVP).

\begin{figure}[!htbp]
    \centering
    \lstset{breaklines,style=MyCustomC,backgroundcolor=\color{lightgrey}}
        \begin{lstlisting}
int main (int argc, char** argv){
  char* input = argv[1];

  char chr = *input; // inserted by obfuscation

  switch (chr) {  // inserted by obfuscation
    case 1:  ...  // original code
      break; 
    case 2:  // obfuscated version of case 1
      break;
    ... 
    default: // another obfuscated version of case 1 
      break;
  }       
  return (*input >= 100);
}
    \end{lstlisting}
        \caption{\obbanescu obfuscation}
    \label{fig:banescu-obf}
\end{figure}

\paragraph{\obbanescu \cite{BanescuCGNP16}} \obbanescu is an anti-symbolic execution 
obfuscation proposed by Banescu et al..
Branch conditions are inserted in basic blocks to divide the input value range
into multiple sets. The code inside each branch of the conditional statement is
an obfuscated version of the original code.
 We distinguish two cases, depending on whether the  branch condition uses a  \lstinline{switch} or a 
\lstinline{if} statement.
{\it In the remaining part of this paper, \obsplit\ will denote the \obbanescu obfuscation with \lstinline{if} statement, and
 \obbanescu  the \obbanescu obfuscation with \lstinline{switch} statement. 
}

The \obbanescu (\lstinline{switch}) scheme introduces an exhaustive \lstinline{switch}
statement over all possible values of a given variable -- see example in \cref{fig:banescu-obf},  
thus yielding  $2^S$ extra-paths, with $S$ the bit size of the variable. 
This scheme enjoys the SVP property as  in each branch of the \lstinline{switch} the target variable can have only one value, and  
it  is also tractable in time provided the \lstinline{switch} is efficiently
compiled into a binary search tree or a jump table, as usual. 
Yet, while not pointed out by Banescu et al., {\it this scheme is not tractable
  in space} (code size) as it leads to {\it huge} amount of code duplication -- the
byte case may be fine, but not above.

\begin{figure}[!htbp]
    \centering
    \lstset{breaklines,style=MyCustomC,backgroundcolor=\color{lightgrey}}
        \begin{lstlisting}
int main (int argc, char** argv){
  char* input = argv[1];

  char chr = *input; // inserted by obfuscation

  if (chr < 30) {  // inserted by obfuscation
      ...          // original code <@$\mathcal{O}$@> 
   }  
   else  ...      // obfuscated version of <@$\mathcal{O}$@>

  return (*input >= 100);
}
    \end{lstlisting}
        \caption{\obsplit obfuscation}
    \label{fig:banescu-obsplit}
\end{figure}

\paragraph{\obsplit \cite{BanescuCGNP16}}  This transformation (\cref{fig:banescu-obsplit}) is similar to
\obbanescu, but the control-flow is split by a condition triggered by a variable. 
This protection is tractable in both time (only one additional condition check per forking point) 
and space (only one block duplication per forking point). Yet, the protection is not SVP.     

\label{subsec:new-patterns}

\paragraph{\obfor (new)} The \obfor scheme (\cref{fig:obf-eg})  replaces assignments \lstinline{ch := chr} by loops  
\lstinline{ch = 0; for (i = 0; i <= chr; i++) ch++;} where \lstinline{chr} is an input-dependent variable.   
Intuitively, such \lstinline{for} loops can be unrolled $n$ times, for any value $n$ that \lstinline{chr} can take at runtime. 
Hence, a loop controlled by a variable 
defined over a bit size $S$   generates up to $2^{S}$ extra-paths, with additional path length of 
   $2^S$. While the achieved protection is excellent, it is {\it intractable}  when 
$S=32$ or $S=64$.  
{\bf The  trick}  is to restrict this scheme to byte-size variables, and then chain such forking points on each byte of the 
variable of interest.  Indeed, \obfor over a byte-size variable generates up to $2^8$ additional paths with an additional path length 
at most of  $2^8$. More informaion about this is given in \cref{sec:annexe-write}. 

Chaining $k$ forking points such \obfor loops  leads up  
to ${2^{8k}}$ extra-paths with an extra-length of only $k \cdot 2^8$, {\it keeping strength while making runtime overhead reasonable}. 
(More precisely with a constant time overhead wrt inputs.)

\paragraph{\obwrite (new)} The \obwrite obfuscation adds self-modification
operations to the code. It replaces an assignment \lstinline{a := input} with a non
input-dependent operation \lstinline{a := k} (with \lstinline{k} an arbitrary constant value) 
and replaces at runtime this instruction by \lstinline{a := i} where \lstinline{i} is the runtime value of \lstinline{input}  
(self-modification). This is illustrated in \cref{fig:wobf}, where the offset move at label \lstinline{L1}
actually rewrites the constant \lstinline{0} at \lstinline{L2} to the value contained at the address of the
input. More detailed information about \obwrite can be found in \cref{sec:annexe-write}.

\begin{figure}[!htbp]
    \centering
    \begin{minipage}[]{0.4\linewidth}
\begin{lstlisting}[backgroundcolor=\color{lightgrey},basicstyle=\scriptsize]

L: mov [a], [input]
\end{lstlisting}
    \end{minipage}~
    \large $\Rightarrow$~
    \begin{minipage}{0.47\linewidth}
\begin{lstlisting}[backgroundcolor=\color{lightgrey},basicstyle=\scriptsize]
<@$~$@> L1: mov L2+off,[input]
<@$~$@> L2: mov [a], 0
\end{lstlisting} 
    \end{minipage}
    \caption{\obwrite obfuscation}
    \label{fig:wobf}
\end{figure}

Symbolic execution engines are not likely to relate \lstinline{a} and \lstinline{input},
thus thinking that \lstinline{a} is constant across all executions. If the dynamic part of the engine spots that
\lstinline{a} may have different values, it will  iterate over every possible values
of \lstinline{input}, creating new paths each time.  The scheme is SVP, and 
its overhead is negligible (2 additional instructions, independent of the bit size of the targeted variable as long as it can be handled 
natively by the underlying architecture). \obwrite has yet two drawbacks: it can be
spotted by a dynamic analysis and needs the section containing the code to be writable. 

\begin{table}[!htbp]
    \centering
    \caption{Classification of obfuscation schemes}
    \label{tab:obf-properties}

    \vspace{0.1cm}

    \resizebox{\columnwidth}{!}{%
    \begin{tabular}{|ll|c|c|c||c|}
        \cline{3-6}
        \multicolumn{1}{c}{}&& \multicolumn{2}{c|}{Tractable} & SVP & Stealth \\
        \multicolumn{1}{c}{}&& Time & Space & & \\
        \hline
        {\obbanescu \cite{BanescuCGNP16}} & \textsf{switch} & \gooddef/ & \baddef/ & \gooddef/ & \baddef/  \\
          {\obsplit \cite{BanescuCGNP16}} & \textsf{if}  & \gooddef/ & \gooddef/ & \baddef/ & \gooddef/ \\
        \hline
        \multirow{2}{*}{\obfor} & \textsf{word} & \baddef/ & \gooddef/ & \gooddef/ & \gooddef/ \\
        & \textsf{byte} & \gooddef/ & \gooddef/ & \gooddef/ & \gooddef/ \\
        \hline
        \multicolumn{2}{|l|}{\obwrite} & \gooddef/ & \gooddef/ & \gooddef/ & $\sim$ \\
        \hline
    \end{tabular}}
\end{table}

\paragraph{Summary} 
\label{sec:schemes:summary}
The properties of every scheme presented so far are summarized in \cref{tab:obf-properties} -- stealth is discussed in \cref{sec:stealth}.    
{\it Obfuscation schemes from the literature are not fully satisfactory}: \obbanescu  is 
space expensive  and easy to spot,   \obsplit is not strong enough  (not SVP).  
On the other hand, {\it the new schemes} \obfor (at byte-level) and \obwrite\  
{\it are both strong (they satisfy SVP) and tractable, making them perfect anti-DSE protections}.  

As a consequence, we suggest using variations of \obfor as the main protection
layer, with \obwrite deployed only when self-modification and unpacking are
already used (so that the scheme remains hidden). \obbanescu can be used {\it
  occasionally} but only on byte variables to mitigate space explosion. \obsplit
can add further diversity and code complexity, but it should not be the main
defense line.  All these protections must  be inserted in a {\it
  resistant-by-design} manner (\cref{subsec:resbydesign}) together with {\it
  diversity of implementation} (\cref{sec:attack:pattern}). 

\section{Anchorage policy}
\label{sec:anchorage-policy}

We need to ensure that inserting path-oriented protections into a program gives
real protection against DSE and will not be circumvented by attackers. 

\subsection{Optimal composition}
\label{sec:optimal-composition}

We show how to  combine the
 forking points in order to obtain  \emph{strong} obfuscation schemes. 
The issue with obfuscation scheme combination 
is that some forking points could hinder the efficiency of other 
forking points -- imagine a {\tt if($x \geq 100$)} split followed by a {\tt if($x \leq 10$)} split: 
we will end up with 3 {\it feasible} paths rather than the expected $2\times2 = 4$, 
as one of the path if infeasible ($x > 100 \wedge x \leq 10$).

Intuitively, we would like the forking points to be independent from each other, in the sense 
that their efficiency combine perfectly 

\begin{defi}[Independence]
Let us  consider a program $\mathcal{P}$ and $\sigma$ a path
of $\mathcal{P}$. We obfuscate this program alternatively with two forking points $\mathcal{F}_1$ and 
$\mathcal{F}_2$ such that $\sigma$ encounters both forking points. This results 
in three obfuscated programs: $\mathcal{P}_1$, $\mathcal{P}_2$ and $\mathcal{P}_{1,2}$.
We note $\#\sigma_1$ ({\it resp.} $\#\sigma_2$) the set of feasible paths 
created from $\sigma$ when encountering only $\mathcal{F}_1$ in $\mathcal{P}_1$ 
({\it resp.} $\mathcal{F}_2$ in $\mathcal{P}_2$) 
and $\#\sigma_{1,2}$ the set of feasible paths created from $\sigma$ when encountering
both $\mathcal{F}_1$ and $\mathcal{F}_2$ in $\mathcal{P}_{1,2}$. 
    $\mathcal{F}_1$ and $\mathcal{F}_2$ are independent over a program $\mathcal{P}$ 
    if for all path $\sigma$ passing through $\mathcal{F}_1$ and $\mathcal{F}_2$:
    $$\#\sigma_{1,2} = \#\sigma_1 \times \#\sigma_2$$
\end{defi}

An easy way to obtain forking point independence is to consider forking points 
built on independent variables -- variables are {\it independent} if their values 
are not computed from the same input values.  
Actually, if independent forking points are well placed in
the program, path-oriented protections ensure an exponential increase in the number of paths.

\begin{theo}[Optimal Composition]
    \label{theo:opti-comp}
    Suppose that $\mathcal{P}'$  is obtained by obfuscating the program $\mathcal{P}$. 
    If every original path of $\mathcal{P}$ goes through at least $k$ independent 
    forking points of $\mathcal{P}'$ inserting at least $\theta$ feasible paths, then
    $\#\Pi_{\mathcal{P}'} \geq \#\Pi_\mathcal{P} \cdot \theta ^k$
\end{theo}

\begin{proof}
Let's consider a program $\mathcal{P}$ with $\#\Pi$ original paths $\sigma_i$, $i \in \{1\dots\#\Pi$\}.
We obfuscate $\mathcal{P}$ into $\mathcal{P'}$ with an obfuscation scheme adding $n$ {\em independent} 
forking points inserting $\#\sigma_{1..n}$ feasible paths.
The forking points  are placed such that every original path now contains at least $k$ forking points.

\begin{itemize}
   \item The total number of paths of $\mathcal{P'}$ is:
        $$\#\Pi_{\mathcal{P'}} = \sum_{\sigma_i}\#\sigma_i, \quad \sigma_i \in \Pi_\mathcal{P}$$
    \item According to the definition of independence, one original path $\sigma_i$ 
        with {\bf at least} k forking points inserting $\#\sigma_{i,\{1..k\}}$ feasible paths
        creates $\#\sigma_i \geq \prod_{j=1}^k\#\sigma_{i,j}$ new paths
    \item Then,
        $$\#\Pi_{\mathcal{P'}} \geq \sum_{\sigma_i}(\prod_{j=1}^k\#\sigma_{i,j}), 
        \quad \sigma_i \in \Pi_\mathcal{P}$$
        We write $\theta = min_{i,j}(\#\sigma_{i,j})$
        $$\#\Pi_{\mathcal{P'}} \geq \sum_{\sigma_i}(\theta^k), \quad \sigma_i \in \Pi_\mathcal{P}$$
        $$\#\Pi_{\mathcal{P'}} \geq \#\Pi_\mathcal{P}\times\theta^k$$

\end{itemize}

\end{proof}

By choosing enough independent SVP forking points (one for each input variable), we can even 
 ensure that DSE will have to enumerate over all possible input values of the program under analysis, 
hence performing as bad as mere  \emph{brute forcing}.    

\paragraph{Implementation} Ensuring that each path will go through at least $k$ forking points 
can be achieved by carefully selecting the points in the code where the forking points are
inserted: a control flow graph analysis provides information about where and how many forking
points are needed to cover all paths. The easiest way to impact all paths at once is to 
select points in the code that are not influenced by any conditional statement. 
Dataflow analysis can be used further in order to ensure that 
the selected variables do not share dependencies with the same input (independent variables).  

\subsection{Resistance-by-design to taint and slice}
\label{subsec:resbydesign}
\label{sec:resbydesign}

Taint analysis \cite{SchwartzAB10} and (backward) slicing \cite{SrinivasanR16} are two advanced code simplification 
methods built on the notion of {\it data flow relations} through 
a program.  
These data flow relations can be defined as \emph{Definition-Use} (Def-Use)
chains -- as used in compilers. Data are {\it defined} when variables are
assigned values or declared, and {\it used} in expressions.
Taint  (resp. Slice) uses Def-Use chains  
to replace input-independent by its constant effect (resp. remove code not
impacting the output).  
If there exists a Def-Use chain linking data $x$ to data $y$, we write:
$x$ \chain{} $y$. 
\emph{Relevant} variables are defined as having both a Def-Use
chain with an input and one with an output:
\begin{defi}[Relevant Variable]
   $x$ is \textbf{relevant} if there exists at least two Def-Use chains
   such that input \chain{} $x$ and $x$ \chain{} output.
\end{defi}
\smallskip 

A {\it sound} taint analysis ({\it resp.} slice analysis) marks {\it at least} all variables
(x,a) such that $input$ \chain{} $(x,a)$ ({\it resp.} $(x,a)$ \chain{} $output$). 
Unmarked variables are then safely removed (slicing) or set to their constant value (tainting).  
Thus, in order to  resist  by-design to such attacks, a  protection must 
rely on code that will be marked by both slicing and tainting.   

Here, we refine the definition of a forking point $\mathcal{F}$: it can be viewed 
as two parts, a guard $\mathcal{G}$ --- the condition --- and an action $\mathcal{A}$ 
--- the code in the statement. 
We denote by $Var(\mathcal{F})$ the set of variables in $\mathcal{G}$ and $\mathcal{A}$.  
We say that $\mathcal{F}$ is {\it built upon relevant variables} if all variables in $Var(\mathcal{F})$ are relevant. 

\begin{theo}[Resistance by design] Let us consider a program $P$ and a forking point 
    $\mathcal{F}$. Assuming $\mathcal{F}$ is built upon relevant variables, then $\mathcal{F}$  is slice and taint resistant.  
\end{theo}

\begin{proof}
By definition, a sound taint analysis $\mathcal{A_T}$ will mark any relevant variable 
(as they depend from input).  
Hence, if $\mathcal{F}$ is built upon relevant variables, then all variables 
$v \in Var(\mathcal{F})$ will be marked 
by $\mathcal{A}$, hence taint analysis $\mathcal{A}$ will yield no simplification on $\mathcal{F}$.  
In the same manner,  a sound slice analysis $\mathcal{A_S}$ will mark any relevant variable 
(as they impact the output), 
implying that  if $\mathcal{F}$ is built upon relevant variables, then analysis $\mathcal{A_S}$ 
will yield no simplification on $\mathcal{F}$. 
\end{proof}

{\it We can actually note that to obtain resitance by-design: (1) against tainting, it is
    sufficient to have input \chain{} $var(\mathcal{G})$, and (2) against slicing, it is sufficient to have  $var(\mathcal{A})$ 
    \chain{} output because if a slicing analysis is able to remove the action $\mathcal{A}$ of the foking point then all branches
    can be simplified altogether.}

\paragraph{Implementation} Relevant variables can be identified in at least two ways. 
First, one can modify standard compiler algorithms computing  {\it possible} 
Def-Use chains in order to compute {\it real} Def-Use chains -- technically, going from a 
{\it may analysis} to a {\it must analysis}. 
A more original solution observes at runtime  a set of real Def-Use chains 
and deduces accordingly a set of relevant variables. This method does not require any 
advanced static analysis, only basic dynamic tracing features.   

\section{Threats}
\label{sec:threats}

In this section we discuss possible threats to path-oriented protections and propose
adequate mitigations. Indeed, when weaving the forking points within the code of a program, we need to
ensure that they are hard to discover or remove. 
Three main attacks appear
to seem effective against path-oriented protections:
\begin{enumerate*}
\item \label{thr:taint} taint analysis, 
\item \label{thr:slicing} backward slicing,  
\item and pattern attacks. 
\end{enumerate*}
We showed how path-oriented protections can be made {\it resistant by-design} to
Taint and Slice in \cref{sec:resbydesign}. 
We discuss now pattern attacks, as well as  stealth issues and the specific
unfriendly case of programs with a small input space.

\begin{figure}[!htbp]
  \medskip
 \begin{enumerate}[label=\Large\protect\textcircled{\small\arabic*}]
 \item
   \begin{minipage}[t]{0.99\linewidth}
\vspace*{-.5cm}
\begin{lstlisting}[backgroundcolor=\color{lightgrey}]
for (int i = 0; i++; i < input) a++;
\end{lstlisting}
     \end{minipage}
     
   \item
     \begin{minipage}[t]{0.99\linewidth}
\vspace*{-.45cm}
\begin{lstlisting}[mathescape,backgroundcolor=\color{lightgrey}]
for (int i = 0; i++; i < input) 
  a = (a ^ 1) + 2 * (a & 1);
\end{lstlisting}
     \end{minipage}
   \item
     \begin{minipage}[t]{0.99\linewidth}
\vspace*{-.4cm}
\begin{lstlisting}[backgroundcolor=\color{lightgrey}]
int i = 0;
while (i < input) { i++; a++; }
\end{lstlisting}
     \end{minipage}
            
   \item
     \begin{minipage}[t]{0.99\linewidth}
\vspace*{-.4cm}
\begin{lstlisting}[backgroundcolor=\color{lightgrey}]
int f(int x) { 
  return (x <= 0 ? 0 : f(x - 1) + 1);
}

a = f(input);
\end{lstlisting}
     \end{minipage}

   \item
     \begin{minipage}[t]{0.99\linewidth}
\vspace*{-.4cm}
\begin{lstlisting}[backgroundcolor=\color{lightgrey}]
#define A // arbitrary value

int f(int x) {
  return x <= 0 ? 0 : A + g(x - 1);
} 

int g(int x) { 
  return !x ? 1 - A : 2 - A + f(--x);
}

a = f(input);
\end{lstlisting} 
     \end{minipage}
   \end{enumerate}
 \tasseup
   \caption{Several encodings of protection \obfor}
   \label{fig:pattern-examples}
\end{figure}

\subsection{Pattern attacks}
\label{subsec:attack:pattern}  
\label{sec:attack:pattern}

{\it Pattern attacks} search for specific patterns in the code of
a program to identify, and remove, known obfuscations. This kind of
analysis assumes  more or less  similar 
constructions across all implementations of an obfuscation scheme. 
A general defense against pattern attacks is {\it diversity}. It works 
well in the case of path-oriented protections: on the one hand the schemes we 
provide can be implemented in many ways,  
and on the other hand our framework provides guidelines to design new schemes 
-- altogether, it should be enough to defeat pattern attacks.     
Regarding diversity of implementations,  
the standard \obfor forking point  can be for example replaced by a \lstinline{while} 
loop,  (mutually) recursive  function(s), 
the loop body can be masked through MBAs, etc. These variants can be combined as in
\cref{fig:pattern-examples},  
and we can imagine many other variations. 

The other schemes as well can  be implemented in many ways, and 
we could  also think of ROP-based encoding~\cite{DBLP:conf/ccs/Shacham07} 
or  other  diversification techniques.  
Altogether, it should provide a powerful enough mitigation against pattern attacks. 

\subsection{Stealth}
\label{sec:stealth}

In general, code protections are better when hard to identify, in order to
prevent human-level attacks like stubbing parts of the code or designing targeted
methods. Let us evaluate the stealthiness of path-oriented protections (summary in \cref{tab:obf-properties}). 
\obsplit and \obfor do not use rare operators or exotic control-flow structures,
only some additional conditions and loops scattered through the program. Hence
\obsplit and \obfor are  considered hard to detect on binary
code, though \obfor\ especially may be visible at source level. 
\obbanescu is easy to spot at source level: \lstinline{switch} statements with
hundreds of branches are indeed distinctive.
Compilation makes it harder to find but the height of the produced binary search
trees or the size of the generated jump table are easily perceptible. 
\obwrite stands somewhere in between. It cannot be easily discovered
statically, but is trivial  to detect dynamically. However, since self-modification and
unpacking are common in obfuscated codes,  \obwrite\ could well be
mistaken for one of these more standard (and less damaging) protections.

\subsection{Beware: programs with  small input space}
\label{subsec:limits}

Resistance by design (\cref{sec:resbydesign}) relies on \emph{relevant variables}, 
so we only have limited room for forking points.
In practice it should not be problematic as \cref{sec:case-study}
shows that we already get very strong protection with only 3 input bytes --
assuming a SVP scheme. 
Yet, for programs with very limited input space, we may 
need to add {\it (fake) crafted inputs} for the input space to become (apparently) larger -- see \obsplit 
example in Fig.~\ref{fig:obf-eg}.    
In this case, our technique still ensures resistance against tainting
attacks, but slicing attacks may now succeed. 
The defender must then rely on well-known (but imperfect) anti-slicing protections   
to blur   code analysis  through hard-to-reason-about constructs, such as 
pointer aliasing, arithmetic and bit-level identities, etc.  

\section{Experimental evaluation}
\label{sec:case-study}
\label{sec:xp-eval}

The experiments below seek to answer four Research Questions\footnote{
Download at \url{https://bit.ly/2wYSEDG} -- deanonymation}:  

\begin{enumerate}[label={\bf RQ\arabic*},ref={\bf RQ\arabic*}]

\item \label{rq1} What is the impact of path-oriented protections on semantic attackers? 
Especially, we 
consider DSE attack and two different attacker goals: Path Exploration \textbf{(Goal 1)} and Secret Finding \textbf{(Goal 2)}. 

\item \label{rq2}
  What is the cost of path-oriented protections for the defender in  runtime overhead and 
  code size increase?  

\item \label{rq3}
  What is the effectiveness of our resistance-by-design  mechanism against  taint and slice attacks? 

\item \label{rq4} What
  is the difference between standard protections, path-oriented protections and SVP protections?

\end{enumerate}

\subsection{Experimental setup}
\label{sec:implem}
\label{sec:implementation}

\paragraph{Tools}  
Our attacker mainly comprises the state-of-the-art source-level DSE tool 
\klee  (version 1.4.0.0 with LLVM 3.4, POSIX runtime and  STP
solver). \klee is highly optimized \cite{DBLP:conf/osdi/CadarDE08} and works from source code, so {\it it is arguably the worst case DSE-attacker 
we can face} \cite{BanescuCGNP16}. We have considered standard search heuristics
(DFS, BFS, Non-Uniform Random Search) but report only about DFS 
as others perform slightly worse ({\it see Appendix}). Also, we used two other {\it binary-level} DSE tools,  
\binsec \cite{DBLP:conf/wcre/DavidBTMFPM16} and \triton \cite{Triton}, with similar results. 

Regarding standard defense, we use Tigress \cite{CollbergMMN12}, a freely available state-of-the-art obfuscator featuring many standard obfuscations and allowing to precisely control which ones are used -- making Tigress 
a tool of choice for the systematic evaluation of deobfuscation methods
\cite{DBLP:conf/sp/BardinDM17,SalwanBarPot18,BanescuCGNP16}.  

\paragraph{Protections} 
We only consider tractable path-oriented protections  
and select both a new SVP scheme (\obfor) and an existing non-SVP scheme (\obsplit), inserted in a 
robust-by-design way. 
We  vary the number of forking points per path (parameter $k$). 

We also consider standard protections:  {\it Virtualization} (up to 3 levels), 
{\it arithmetic encoding} and {\it flattening}~\cite{WangHKD01}.   Previous work
\cite{BanescuCGNP16} has shown that nested virtualization is the 
sole standard protection useful against DSE. Our results confirm that, so we report only results about virtualization {\it (other results  partly in Appendix)}. 

\subsection{Datasets}
\label{sec:datasets}

{\it We select small and medium  programs for experiments as they represent the
 worst case for program protection. If path-oriented protections can slow down 
 DSE analysis substantially on smaller 
 programs, then those protections can only give better results for larger programs.}    

\paragraph{Dataset \#1}
This synthetic dataset from Banescu \textit{et al}
\footnote{\url{https://github.com/tum-i22/obfuscation-benchmarks}} \cite{BanescuCGNP16}  offers a valuable diversity of functions
and has already been used to assess resilience against DSE. 
It has 48 C programs (between 11 and 24 lines of code)  including
control-flow statements, integer arithmetic and system calls to \lstinline{printf}. 
We exclude 2 programs because reaching full coverage took
considerably longer than for the other 46 programs  and blurred the overall results. 
Also, some programs have only a 1-byte input space, making them too easy to brute force (Goal 2).  
We turn them into equivalent 8-byte input programs with same number of paths -- additional input are not used by the program, 
but latter protections will rely on them. 
The maximum time to obtain full coverage on the {\bf 46} programs with \klee is $33s$, mean time  is $2.34s$ (\cref{sec:dataset:benchmark}). 

\paragraph{Dataset \#2} 
The second dataset comprises 7 larger realistic programs, representative of real-life 
protection scenarios: 4 hash functions (\textsf{City}, 
\textsf{Fast}, \textsf{Spooky}, \textsf{md5}), 2 cryptographic encoding functions 
(\textsf{AES}, \textsf{DES}) and a snippet from the \textsf{GRUB} bootloader. 
Unobfuscated programs have between 101 and 934 LOCs:
\klee needs at most  $33.31s$ to explore all paths, mean time is $8s$  
(\cref{sec:dataset:benchmark}). 

\subsection{Impact on Dynamic Symbolic Execution}
\label{sec:impact-dse}
\label{sec:path-exploration}
\label{sec:secret-finding}

\paragraph{Protocol} To assess the impact of protections against DSE, we
consider the induced {\it slowdown} (time) of symbolic execution on an
obfuscated program w.r.t. its original version. Fore more readable results, we only
report whether DSE achieves its goal or times out.

For Path Exploration (\textbf{Goal 1}), we use programs from {Datasets \#1} and
{\#2}, add the protections and launch \klee untils 
it reports full coverage or times out --  3h for { Dataset \#1}, or a
5,400$\times$ average slowdown, 24h for {Dataset \#2}, or  a  10,000x average slowdown. 

For Secret finding (\textbf{Goal 2}), we  modify the programs from both datasets into ``secret finding'' 
oriented code (e.g.,\textit{win / lose}) and set up \klee\ to stop execution as soon as the secret is found. 
We take the whole Dataset \#2, but restrict  Dataset \#1 to the 15 programs with 16-byte input space. 
We set smaller timeouts (1h for { Dataset \#1}, 3h and 8h for { Dataset \#2})
as the time to find the secret with \klee on the original programs
is substantially lower ($0.3s$ average).

\begin{table}[!htbp]
    \centering
    \caption{Impact of obfuscations on DSE}
    \label{tab:impact-dse}
        \resizebox{\columnwidth}{!}{%
            \begin{tabular}{|l|C{1.2cm}|C{1.2cm}|C{1.2cm}|C{1.2cm}|C{1.2cm}|}
        \hline
        \multirow{2}{*}{\textbf{Transformation}} & \multicolumn{2}{c|}{\textbf{Dataset \#1}} &
        \multicolumn{3}{c|}{\textbf{Dataset \#2}} \\
        \cline{2-6}
	&&&&&\\[-0.9em]
        (\#TO/\#Samples) & \itshape Goal 1 \newline 3h TO & \itshape Goal 2 \newline 1h TO & \itshape Goal 1
        \newline 24h TO & \itshape Goal 2 \newline 3h TO & \itshape Goal 2 \newline 8h TO \\
        \hline\hline
        \textbf{Virt} & $0/46$ & $0/15$ & $0/7$ & $0/7$ & $0/7$ \\
        \textbf{Virt $\times2$} & $1/46$ & $0/15$ & $0/7$ & $0/7$ & $0/7$ \\
        \textbf{Virt $\times3$} & $5/46$ & $2/15$ & $1/7$ & $0/7$ & $0/7$ \\
        \hline\hline
        \textbf{SPLIT ($k=10$)} & $1/46$ & $0/15$ & $0/7$ & $0/7$ & $0/7$ \\
        \textbf{SPLIT ($k=13$)} & $4/46$ & $0/15$ & $1/7$ & $1/7$ & $0/7$ \\
        \textbf{SPLIT ($k=17$)} & $18/46$ & $2/15$ & $3/7$ & $2/7$ & $1/7$ \\
        \textbf{FOR ($k=1$)} & $2/46$ & $0/15$ & $0/7$ & $0/7$ & $0/7$ \\
        \textbf{FOR ($k=3$)} & $30/46$ & $8/15$  & $3/7$ & $2/7$ & $1/7$ \\
        \textbf{FOR ($k=5$)} & $46/46$ & $15/15$ & $7/7$ & $7/7$ & $7/7$ \\
        \hline
    \end{tabular}} 
\end{table}

\paragraph{Results \& Observations} \cref{tab:impact-dse} shows the number of timeouts during symbolic execution for each obfuscation
and goal. For example, \klee is always able to cover all paths on Dataset \#1  against simple Virtualization (0/46 TO) -- the protection is useless here,  
while it fails on $\approx 40\%$ of the programs with  \obsplit  ($k=17$), and never succeeds with  \obfor  ($k=5$).   

As expected, higher levels of protections (more virtualization layers or more forking points) 
result in better protection.
Yet, results of \cref{sec:impact-perf} will show that while increasing forking points is cheap, increasing levels of virtualization is quickly prohibitive. 

Virtualization is rather weak for both goals  -- only 3 levels of virtualization 
manage some protection.  
\obfor performs very well for both goals: with $k=3$ and Dataset \#1, \obfor induces a timeout for
more than half the programs for both goals, which is significantly better than {\bf Virt}$\times3$. With $k=5$, all programs timeout.
In between, \obsplit is efficient for Goal 1 (even though it requires much higher $k$ than \obfor) but not for Goal 2 -- see for example Dataset \#1 and $k=17$: 
39\% timeouts (18/46) for Goal 1, only 13\% (2/15) for Goal 2. 

\paragraph{Other (unreported) results} 
All standard protections from Tigress we used turns out to be ineffective against DSE -- for example  \textit{Flattening} and 
 \textit{EncodeArithmetic} on Dataset\#1 slows path exploration by a maximum factor of 10, 
which is far from critical. 
Search heuristics obviously do not make any difference in the case of Path
Exploration (Goal 1). Still, DFS tends to perform slightly better than BFS and
NURS against \obsplit in the case of Secret Finding (Goal 2). No
other difference is visible.
Experiments with   two binary-level DSE engines supported by different SMT solvers  
(\binsec \cite{DBLP:conf/wcre/DavidBTMFPM16} with  
Boolector\cite{DBLP:conf/tacas/BrummayerB09}, and \triton \cite{Triton} with  Z3\cite{DBLP:conf/tacas/MouraB08}) 
are in line with those reported here. Actually, as expected, both engines perform worse than \klee.  
Part of these results can be found in Appendix. 

\paragraph{Conclusion} As already stated in the literature, standard protections such as nested virtualization are mostly inefficient against DSE attacks. 
Path-oriented protections are shown here to offer a stronger protection. Yet, care must be taken. 
Non-SVP path protections such as \obsplit\ do improve over nested virtualization 
(\obsplit\ with $k=13$ is roughly equivalent to \textbf{Virt $\times3$}, with $k=17$ it is clearly superior), but they provide only a weak-to-mild protection  
 in the cases of Secret Finding (Goal 2) or large time outs.   
On the other hand, SVP protections (represented here by \obfor) 
are able to discard all DSE attacks on our benchmarks for both Path Exploration 
and Secret Finding with only $k=5$,
demonstrating a protection power
against DSE far above those of standard and  non-SVP path protections. 

To conclude, path-oriented protections are indeed a tool of choice against
DSE attacks (\ref{rq1}), much stronger than standard ones
(\ref{rq4}). In addition, SVP allows to predict the strength difference of
these protections (\ref{rq4}), against Coverage or Secret Finding.

\subsection{Impact on Runtime Performance} 
\label{sec:impact-perf}

\paragraph{Protocol} We evaluate the cost of path-oriented protections by measuring
the {\it runtime overhead} (RO)
and the (binary-level) {\it code size increase} (CI) of an obfuscated program w.r.t. its original version.
We  consider also two variants of \obfor -- its recursive encoding REC (\cref{sec:attack:pattern}) and the more robust P2 encoding (\cref{sec:xp:robust}), as well as  the untractable  word-level 
\obfor scheme (\cref{sec:strong-schemes-forking-points}), coined WORD.  

\begin{table}[!htbp]
    \centering
    \caption{Impact of obfuscations on runtime performance}
    \label{tab:impact-perf}
    \resizebox{\columnwidth}{!}{%
        \begin{tabular}{|l|R{1.4cm}|R{1.2cm}|R{1.4cm}|R{1.2cm}|}
        \hline
        \multirow{2}{*}{\textbf{Transformation}} & \multicolumn{2}{c|}{\textbf{Dataset \#1}} &
        \multicolumn{2}{c|}{\textbf{Dataset \#2}} \\
        \cline{2-5}
        &&&&\\[-0.9em]
        & \multicolumn{1}{c|}{\itshape RO} & \multicolumn{1}{c|}{\itshape CI} & \multicolumn{1}{c|}{\itshape RO} & \multicolumn{1}{c|}{\itshape CI} \\
        \hline\hline
        &&&&\\[-0.9em]
        \textbf{Virt} & $\times1.5$ & $\times1.5$ & $\times1.7$ & $\times1.4$ \\
        \textbf{Virt $\times2$} & $\times15$ & $\times2.5$ & $\times5.1$ & $\times2.1$ \\
        \textbf{Virt $\times3$} & $\times1.6\cdot10^3$ & $\times4$ & $\times362$ & $\times3.0$ \\
        \hline\hline
        &&&&\\[-0.9em]
        \textbf{SPLIT ($k=10$)} & $\times1.2$ & $\times1.0$ & $\times1.0$ & $\times1.0$ \\
        \textbf{SPLIT ($k=13$)} & $\times1.2$ & $\times1.0$ & $\times1.0$ & $\times1.0$ \\
        \textbf{SPLIT ($k=50$)} & $\times1.5$ & $\times1.5$ & $\times1.1$ & $\times1.0$ \\
        \textbf{FOR ($k=1$)} & $\times1.0$ & $\times1.0$ & $\times1.0$ & $\times1.0$ \\
        \textbf{FOR ($k=3$)} & $\times1.1$ & $\times1.0$ & $\times1.0$ & $\times1.0$ \\
        \textbf{FOR ($k=5$)} & $\times1.3$ & $\times1.0$ & $\times1.1$ & $\times1.0$ \\
        \textbf{FOR ($k=50$)} & $\times1.5$ & $\times1.5$ & $\times1.2$ & $\times1.1$ \\
        \hline\hline
&&&&\\[-0.9em]
        \textbf{FOR ($k=5$) P2} & $\times1.3$ & $\times1.0$ & $\times1.1$ & $\times1.0$ \\
        \textbf{FOR ($k=5$) REC} & $\times3.0$ & $\times1.0$ & $\times2.7$ & $\times1.0$ \\
        \hline
        &&&&\\[-0.9em]
        \textbf{FOR ($k=1$) WORD} & $\times2.6\cdot10^3$ & $\times1.0$ & $\times2.1\cdot10^3$ & $\times1.0$ \\
        \hline
    \end{tabular}}
\end{table}

\paragraph{Results \& Observations} 
Results  are shown in \cref{tab:impact-perf} as average values over
all programs in the datasets.   
As expected, nested virtualization introduces a significant and prohibitive runtime overhead (three layers: $\times 1.6\cdot10^3$ for Dataset \#1 and $\times 362$ for Dataset \#2), 
and each new layer comes at a high price (from 1 to 2: between $\times 3$ and $\times 10$; from 2 to 3: between $\times 70$ and $\times 100$).   
Moreover, the code size is also increased, but in a more manageable way (still, at least $\times 3$ for three layers). 
On the other hand, \obsplit and \obfor introduce only very low runtime overhead (at most $\times 1.3$ on Dataset \#1 and $\times 1.1$ on Dataset \#2), and no noticeable 
code size increase is reported even for $k=50$.
Regarding variants of \obfor, P2 does not show any overhead w.r.t.~\obfor, while 
the recursive encoding REC comes at a higher price. Finally, {\it as predicted by our framework}, WORD  is  intractable.  

\paragraph{Conclusion}
As expected, tractable path-oriented protections indeed yield only a very slight overhead, both in terms of time or code size (\ref{rq2}),  
and improving the level of protection ($k$) is rather cheap, while nested virtualization comes at a high price (\ref{rq4}).  
Coupled with results of \cref{sec:impact-dse}, it turns out that path-oriented protections offer a much better anti-DSE protection than 
nested virtualization at a runtime cost several orders of magnitude lower. Also, the code size increase due to path-oriented protections seems 
compatible with strict memory requirements (e.g., embedded systems) where it is not the case for nested virtualization.     

\begin{table}[!htbp]
    \caption{Robustness  of path-oriented protections}
    \label{tab:robustness}
    \centering
    \begin{tabular}{|l|C{1cm}|C{2.cm}||C{1cm}|}
        \hline
        & \multicolumn{3}{c|}{}\\[-0.9em]
        \multirow{2}{*}{\textbf{Tool}} & \multicolumn{3}{c|}{\textbf{Robust ?}} \\
        \cline{2-4}
        &&&\\[-0.9em]
        & P1 & P2 & P3 \\ 
        & (basic)   & (obfuscated)   &  (weak)   \\
        \hline
        &&&\\[-0.9em]
        GCC -Ofast & \gooddef/ & \gooddef/ & \baddef/ \\
        clang -Ofast& \baddef/ & \gooddef/ & \baddef/ \\
        \hline\hline
        Frama-C Slice & \gooddef/ & \gooddef/ & \baddef/ \\
        Frama-C Taint & \gooddef/ & \gooddef/ & \gooddef/ \\
        \hline\hline
        \triton (taint) & \gooddef/ & \gooddef/ & \gooddef/ \\
        \klee           & \gooddef/ & \gooddef/ & \gooddef/ \\     
        \hline
    \end{tabular}

\gooddef/: no protection simplified  \hfill  \baddef/: $\ge 1$ protection simplified

\end{table}

\subsection{Robustness to taint and slice attacks}
\label{sec:xp:robust}

\paragraph{Protocol} We consider the heavily optimized compilers Clang \&  GCC (many simplifications including slicing), the industrial-strength Frama-C static code analyzer (both its Taint and Slice plugins together 
with precise interprocedural range analysis)    
as well as \triton (which features tainting)   and \klee.  
We take 6 programs from 
 dataset \#1 (with 16-byte input space) and all programs from dataset \#2. 
 We use the \obfor scheme (k=3)  weaved into the code following our
 robust-by-design method (\cref{sec:resbydesign}).  Actually we consider 3
 variants of the scheme: {\bf P1}, {\bf P2} and {\bf P3}. P1 is the simple
 version of \obfor presented in \cref{fig:pattern-examples}, P2 is a mildly
 obfuscated version (adds a {\tt if} statement always evaluating to
 \texttt{true} in the loop -- opaque predicate) and P3 relies on fake inputs (a
 dangerous construction discussed in \cref{subsec:limits}).
A protection will be said to be {\it simplified} when the number of explored paths for
full coverage is much lower than expected (DSE tools), no protection code is
marked by the analysis tool (Frama-C) or running \klee on the produced code does
not show any difference (compilers).

\paragraph{Results \& Observations} Results in \cref{tab:robustness} confirm our expectations.
No analyzer but \textsf{clang} is able to simplify 
our robust-by-design protections (P1 and P2), whereas the weaker P3 is broken 
by slicing  (\textsf{GCC}, \textsf{clang}, Frama-C)  but not by tainting -- exactly as pointed out in \cref{subsec:limits}. 
Interestingly, \textsf{clang -Ofast} simplifies scheme P1, {\it not due to slicing} (this is resistant by design), 
but thanks to some loop simplification more akin to a pattern attack, relying
on finding an affine relation between variables and loop counters. The slightly obfuscated version P2 
is immune to this particular attack.  

\paragraph{Conclusion} It turns out that our robust-by-design method indeed works as expected against taint and slice (\ref{rq3}). 
Yet, care must be taken to avoid pattern-like simplifications. Note that in a real scenario, the attacker must work 
on binary code, making  static analysis much more complicated. Also,  virtualization, unpacking or self-modification 
can be used in addition to path-oriented protections to completely hinder static analysis. 

\section{Application: hardened benchmark}
\label{sec:application:hardened}

We propose a set of benchmarks containing 4 programs
from Ba\-nes\-cu's dataset and our 6 real-world programs (\textsf{GRUB} excluded) 
from \cref{sec:datasets}  in order to help advance the state of the art 
of symbolic deobfuscation.  
Each program comes with two setups, Path 
Exploration and Secret Finding,  
obfuscated with both 
a path-oriented protection (\obfor k=5,  taint- and slice- resistant) and a 
virtualization layer against human and static attacks\footnote{Sources 
available at \url{https://bit.ly/2GNxNv9}}. 
\cref{tab:hade} shows the performance of \klee, 
\triton and \binsec 
(Secret Finding, $24$h timeout).   
Unprotected  and virtualized codes are easily solved, but hardened 
versions  remain unsolved within the timeout, for every tool.

\begin{table}[!htbp]
\caption{Results on 10 hardened examples (secret finding)} \label{tab:hardened-benchmark}
\label{tab:hade}

\resizebox{\columnwidth}{!}{%
\begin{tabular}{|l|c|c|c|}

\hline 
        &  Unprotected    & Virt $\times1$    &   {\bf Hardened} -- \obfor (k=5)  \\   
        &   (TO = 10 sec) &   (TO = 5 min)              &     \bf (TO = 24h)               \\   

\hline

\klee     &   10/10         &    10/10                    &    {\showmore 0/10} \gooddef/                      \\ 

\binsec   &    10/10         &    10/10                    &   {\showmore 0/10} \gooddef/                      \\ 

\triton   &   10/10         &    10/10                    &   {\showmore 0/10}  \gooddef/                      \\ 

\hline

\end{tabular}}
\end{table}

\section{Discussion}
\label{sec:discussion}

\subsection{On the methodology}

We discuss biases our experimental evaluation could suffer from. 

\paragraph{Metrics}
We add overhead metrics (runtime, code size) to the commonly used ``DSE
slowdown'' measure \cite{BanescuCGNP16,SalwanBarPot18},  
giving a better account of the pros and cons of
obfuscation methods.

\paragraph{Obfuscation techniques \& tools}
We consider  the strongest standard obfuscation methods known against DSE, as
identified in previous systematic studies~\cite{BanescuCGNP16,SalwanBarPot18}.
We restrict ourselves to their implementation in Tigress, a widely respected and freely
available obfuscation tool considered state-of-the-art  
-- studies including Tigress
along packers and protected malware~\cite{DBLP:conf/sp/BardinDM17,YadegariJWD15}
do not report serious 
deficiencies about its protections. 
Anyway, the evaluation of the path-oriented protections is independent
of Tigress.    

\paragraph{DSE engines}
We use 3 symbolic execution engines (mostly \klee, also \binsec and \triton) working on
different program representations (C source, binary), with very similar final
results. Moreover, \klee is a highly respected tool, implementing   
advanced path pruning methods
(\textit{path merging}) and solving strategies. 
It also benefits from {\it code-level optimizations} of Clang 
as it operates on LLVM bitcode.   
Previous work~\cite{BanescuCGNP16}
considers \klee as the {\it worst-case
attackers}, in front of \triton~\cite{Triton} and {\sf
  angr}~\cite{DBLP:conf/sp/Shoshitaishvili16}. 

 \paragraph{Benchmarks}
Our benchmarks include Banescu et al.'s synthetic benchmarks~\cite{BanescuCGNP16},
{\it enriched by  7 larger real-life programs} consisting essentially of hash
functions (a typical software asset one has to protect)~\cite{SalwanBarPot18}. 
We also work both on source and binary code to
add another level of variability. 
As already said, the considered programs are rather small, {\it on purpose}, to embody {\it
  the defender worst case}.
Note that this case still represents real life situations, e.g., protecting  
small critical assets from  {\it targeted} DSE attacks.

\subsection{Generality of path-oriented protections}
\label{sec:generality}

{Path-oriented protections} should be effective
on a larger class of  attacks besides DSE -- actually, all major semantic program
analysis techniques.  
Indeed, all path-unrolling methods will  suffer from  path explosion, including  
Bounded model checking \cite{B09}, backward bounded DSE
\cite{DBLP:conf/sp/BardinDM17}  and abstract interpretation with aggressive
trace partitioning \cite{DBLP:conf/wcre/Kinder12}.  
Model checking based on counter-example guided refinement~\cite{DBLP:conf/popl/HenzingerJMS02} will suffer both from path explosion and
Single Value Path protections -- yielding ineffective refinements in the vein of
\cite{BruniGG18}. Finally, standard abstract interpretation
\cite{DBLP:journals/toplas/BalakrishnanR10} will suffer from significant
precision loss due to the many introduced {\it merge points}
-- anyway purely static techniques cannot currently cope with
self-modification or packing. 

\vspace*{-.223cm}
\subsection{Countermeasures and mitigations}
\label{sec:countermeasures}

We can think of three possible mitigations a DSE attacker could use
against our new defenses:
\begin{enumerate*}
\item remove the protection through tainting and slicing;
\item detect our defenses via pattern attacks and  
\item directly handle the protection through advanced semantic techniques for
  loops. 
\end{enumerate*}
{\it Slicing, tainting} and {\it pattern attacks}, are 
thoroughly discussed in \cref{sec:resbydesign,sec:threats}. 

{\it Advanced program analysis techniques for loops} is a very hot research
topic, still largely open in the case of under-approximation methods such as
DSE.  The best methods for DSE are based on path merging
\cite{DBLP:journals/cacm/AvgerinosRCB16}, but they lack a generalization step
allowing to completely capture loop semantics.  Even though \klee implements
such path merging, it still fails against our protections.  Widening in
abstract interpretation \cite{DBLP:journals/toplas/BalakrishnanR10} 
over-approximates loop semantics, but the result is often very crude: using
such over-approximations inside DSE is still an open question.  Anti-implicit
flow techniques \cite{KingHHJ08, LiuM10} may identify dataflow hidden as
control-flow (it identified for instance a \obfor forking point), yet they do
not recover any precise loop semantics and thus cannot reduce path explosion.

Finally, note that: 
\begin{enumerate*}
\item obfuscation schemes can easily be scattered along several
functions (see alternative \obfor\ encodings in \cref{subsec:attack:pattern}) to
bar expensive but targeted intra-procedural attacks --  attackers will need
 (costly) precise inter-procedural methods,
\item real-life attacks are performed on binary code --  binary-level static analysis is known to be extremely hard 
to get precise; and 
\item static analysis is completely broken 
by packing or self-modification.   
\end{enumerate*}

\vspace*{-.1cm}
\section{Related Work}
\label{sec:related-work}

We have already discsussed obfuscation, symbolic execution and symbolic deobfuscation at length throughout the paper, including  successful
applications of DSE-related techniques to deobfuscation
\cite{CooganLD11,YadegariJWD15,SalwanBarPot18,DBLP:conf/sp/BardinDM17}.
In addition, Schrittwieser et al.~\cite{Schrittwieser:2016}
give an exhaustive survey about program analysis-based deobfuscation, 
while Schwartz et al.~\cite{SchwartzAB10} review
DSE, tainting  and their applications in  security.

\paragraph{Limits of symbolic execution}
Anand et al. \cite{ANAND20131978} describe, in the setting of  automatic testing,    the three major weaknesses
of DSE: \textit{Path explosion}, \textit{Path divergence}
and \textit{Complex constraints}. 
Cadar~\cite{Cadar15} shows that compiler optimizations can sensibly alter
  the performance of a symbolic analyzer like \klee, confirming  
 the folklore knowledge that strong enough compiler
optimizations resemble code obfuscations.  
That said, the performance penalty is far from offering a strong defense 
against symbolic deobfuscation.  

\paragraph{Constraint-based anti-DSE protections} 
Most anti-DSE techniques target the constraint solving engine 
 through hard-to-solve predicates. 
 The impact on symbolic deobfuscation through the complexification of constraints has been studied by
Banescu et al. \cite{BanescuCP17}. 
Biondi et al. \cite{BiondiJLS17} propose an obfuscation
based on \textit{Mixed Boolean-Arithmetic} expressions \cite{ZhouMGJ07} to complexify {\it points-to functions},  
making it harder for solvers to determine the trigger. 
Eyrolles et al. \cite{EyrollesGV16} present a similar
obfuscation 
together with a MBA expression simplifier  based on  pattern
matching and arithmetic simplifications. 
Cryptographic hash functions 
hinder current solvers and can replace MBA \cite{SharifLGL08}.    
In general, formula hardness is difficult to predict, and solving such formulas
is a hot research topic. Though cryptographic functions resist solvers up to
now, promising attempts \cite{DBLP:conf/vstte/NejatiLGCG17} exist. More
importantly, private keys must also be protected against symbolic
attacks, yielding a potentially easier deobfuscation subgoal -- a standard
whitebox cryptography issue.

\paragraph{Other anti-DSE protections} 
Yadegari and Debray \cite{YadegariD15} describe obfuscations 
thwarting standard byte-level taint analysis, possibly resulting in missing
legitimate paths for DSE engines using taint analysis (\triton does, \klee and
\binsec do not).  It can be circumvented in the case of taint-based DSE by
bit-level tainting \cite{YadegariD15}.  \textit{Symbolic Code} combines this
idea with {\it input-dependent trigger-based self modifications}. Here, the
dynamic analysis part of DSE must be able to detect these input-dependent
self-modifications. Solutions exist but must be carefully integrated
\cite{DBLP:conf/sp/BardinDM17,BonfanteFMRST15}. 
Wang et al. \cite{WangMJG11} propose an obfuscation based on mathematical
conjectures in the vein of the Collatz conjecture.  This transformation
increases the number of (symbolic) paths through an input-dependent loop, while 
the conjecture (should) ensure that the loop always converges to the same result. 
Banescu et al. \cite{BanescuCGNP16} propose  an anti-DSE technique   
based on encryption and  proved to be highly effective, but it 
requires some form of secret sharing (the key) and thus falls outside  the strict scope of 
MATE attacks that we consider here.   
Stephens et al.  \cite{StephensYCDS18} 
recently proposed an obfuscation based on covert channels 
(timing, etc.) to hide data flow within invisible states. Current tools 
do not handle correctly this kind of protections.  However, the method ensures only
probabilistic correctness and thus cannot be applied in every context.

\paragraph{Systematic evaluation of anti-DSE techniques} 
Banescu et al.~\cite{BanescuCGNP16} 
set the  ground for the experimental evaluation 
of symbolic deobfuscation techniques. Our own experimental evaluation 
extends and refines their method in several ways: new metrics, 
different DSE settings, larger examples.  
Bruni et al. \cite{BruniGG18} propose a mathematically proven obfuscation against
Abstract Model Checking attacks. 

\vspace*{-.1cm}
\section{Conclusion}
\label{sec:conclusions}

Code obfuscation intends to protect proprietary software assets against attacks such as  
reverse engineering or code tampering. Yet, recently proposed (automated) attacks based on 
symbolic execution (DSE) and semantic reasoning  have shown a great potential against 
traditional obfuscation methods. 
We explore a new class of anti-DSE techniques targeting the very weak
spot of these approaches, namely path exploration. We propose a predictive framework 
for understanding such path-oriented protections,    
and we  propose new lightweight, efficient and  resistant 
obfuscations.
Experimental evaluation indicates that our method critically damages symbolic
deobfuscation while yielding only a very small overhead. 


\bibliographystyle{plain}
\bibliography{references}

\crefalias{table}{apptable}
\makeatletter
\setlength{\@fpsep}{8pt}
\makeatother
\appendix

\clearpage
\section{Additional information on patterns}
\label{sec:annexe-write}

\paragraph{\obwrite}
In this section we give more details about the \obwrite obfuscation. As previously stated in
\cref{sec:strong-schemes-forking-points} this transformation involves self-modification: the code directly modify
the executed instructions at runtime. A step by step exemple of \obwrite is presented in 
\cref{fig:write-explication}.

\begin{figure}[!htbp]
    \caption{Step by Step execution of \obwrite with the runtime value of input being 100}
    \label{fig:write-explication}
    \centering
    \begin{lrbox}{0}
    \begin{tabular}{C{0.2cm}cc}
        \centering
        1 & 
        \begin{lstlisting}
<@$~$@>  L: mov [a], [input]
        \end{lstlisting}
        &
        \begin{lstlisting}
<@$~\color{red}\rightarrow$@> L1: mov L2+off, [input]
<@$~$@>   L2: mov [a], 0
        \end{lstlisting} 
        \\
        &&\\[-0.7em]
        2 &  
        \begin{lstlisting}
<@$~$@>  L: mov [a], [input]
        \end{lstlisting}
        &
        \begin{lstlisting}
<@$~$@>   L1: mov L2+off, [input]
<@$~\color{red}\rightarrow$@> L2: mov [a], <@ \textcolor{red}{runtime\_val}@>
        \end{lstlisting}
        \\
        &&\\[-0.7em]
        3 &  
        \begin{lstlisting}
<@$~$@>  L: mov [a], [input]
        \end{lstlisting}
        &
        \begin{lstlisting} 
<@$~$@>   L1: mov L2+off, [input]
<@$~\color{red}\rightarrow$@> L2: mov [a], <@\textcolor{red}{100}@>
        \end{lstlisting}
        \\
        \end{tabular}
    \end{lrbox}
    \resizebox{\columnwidth}{!}{\usebox0}
\end{figure}

The original code (a simple \texttt{mov}) is replaced by a two-line self-modification.
The instruction at address \texttt{L1} replaces the value ``0'' in line \texttt{L2}
by the runtime value of $input$ (step1). Then, when the program executes the instruction
at \texttt{L2} (step 2) the value given to $a$ is actually the value of $input$, which is $100$
in this exemple.

This transformation does not change the semantics of the program as both code give
$input$'s value to $a$. However, where the original code has one path and one code version
the modified program has one path and one code version for each value of $input$.

\paragraph{\obfor} Here, we further explain how the \obfor pattern is made tractable in time.
As we already stated, using an int variable (or bigger) variable substantially increase
the size of execution traces. To mitigate this side effect we have to apply the \obfor
protection only to \textit{byte} variables. This problem is showned in \cref{fig:for-explication}.

\begin{figure}[!htbp]
    \caption{Impact of variable size to \obfor strength and cost}
    \label{fig:for-explication}
    \centering

    \begin{tabular}{c|c|C{1.5cm}}
        Pattern & Paths & Loop iterations\\
        \hline
        \begin{lstlisting} 
int func (<@\color{red}int x@>) {
  int var = 0;
  for (int i=0; i<x;i++) {
    var++;
  }
  return var;
}
        \end{lstlisting}
        & $2^{32}$ & \color{red} $\leq 2^{32}$ \\
        \hline
        \begin{lrbox}{1}
        \begin{lstlisting}
int func (int x) {
 <@\color{red}char tmp[4] = x;@>
 char var[4] = {0};
 for (int i=0;i<<@\color{red}tmp[0]@>;i++){
  var[0]++;
 }
 for (int i=0;i<<@\color{red}tmp[1]@>;i++){
  var[1]++;
 }
 // same for tmp[2:3]
 return var;
}
    \end{lstlisting}
    \end{lrbox}
    \resizebox{3.4cm}{!}{\usebox1}
    & $2^{32}$ & $\leq 4 \times 2^8$ \\
\end{tabular}
\end{figure}

This trick allows \obfor to offer the same symbolic slowdown at almost no cost ---
none in practice for a real-world program.

\section{Statistics on datasets} \label{sec:dataset:benchmark}

We present  additional  statistics on Dataset  \#1 (\cref{tab:unobf-prog-stat}) and Dataset \#2  (\cref{tab:stat-real-programs}). 
For Dataset \#1, recall that 1-byte input programs from the original dataset from Banescu et al.~\cite{BanescuCGNP16}  are automatically turned into equivalent 8-byte input programs with same number of paths: additional input are not used by the program, 
but latter protections will rely on them. We must do so as they are otherwise too easy to enumerate. 

\begin{table}[h]

  \caption{Statistics on Dataset \#1 (46 programs) }`

  \resizebox{\columnwidth}{!}{%
    \begin{tabular}{|l|C{1.3cm}|C{1.3cm}|C{1.3cm}|C{1.3cm}|}
        \hline
        \multirow{2}{*}{\textbf{Entry size}} & \multicolumn{2}{c|}{\textbf{\#LOC}} & \multicolumn{2}{c|}{\textbf{\klee exec. (s)}} \\
        \cline{2-5}
        &&&&\\[-1em]
        & \textbf{\textit{average}} & \textbf{\textit{StdDev.}} & \textbf{\textit{average}} & \textbf{\textit{StdDev.}} \\
        \hline
        16 bytes & $21$ & $1.9$ & $2.6$s & $6.2$ \\
        \hline
        1 byte (*) & $17$ & $2.2$ & $1.8$s & $6.2$ \\
        \hline
    \end{tabular}}
loc: line of code
    \label{tab:unobf-prog-stat}

\smallskip 

(*) 1-byte input programs are automatically turned into equivalent 8-byte input programs with same number of paths. 
 We report \klee execution time on the modified versions.

\end{table}

\begin{table}[htbp]
  \caption{Statistics on Dataset \#2 (7 programs)}

    \centering

    \begin{tabular}{|l|r|r|}
        \hline
        &&\\[-1em]
        \textbf{Program} & \textbf{locs} & \textbf{\klee exec. (s)} \\
        &&\\[-1em]
        \hline
        City hash & $547$ & $7.41$ \\
        Fast hash & $934$ & $7.74$ \\
        Spooky hash & $625$ & $7.12$ \\
        MD5 hash & $157$ & $33.31$ \\
        \hline
        AES & $571$ & $1.42$ \\
        DES & $424$ & $0.15$ \\
        \hline
        GRUB & $101$ & $0.06$ \\
        \hline
    \end{tabular}
    \label{tab:stat-real-programs}
\end{table}

\section{Additional experiments}
\label{sec:additional-xp}

\paragraph{Search heuristics} Results in \cref{tab:secret-finding3} shows that DSE search heuristics does not impact that much overall results 
(cf.~\cref{tab:impact-dse}).  Depth-first search appears to be slightly better than the two other ones for \obsplit, and non-uniform random search 
  appears to be slightly worse than the two other ones for \obfor. Nothing dramatic yet.  

\begin{table}[!htbp]
    \caption{Impact of search heuristics --  Dataset \#1 -- secret finding -- 1h TO}

     \centering

     \resizebox{\columnwidth}{!}{%
   \begin{tabular}{l|C{1.3cm}C{1.3cm}C{1.3cm}|C{1.3cm}}
\multicolumn{1}{c}{}        & \multicolumn{4}{c}{\textbf{Timeouts}} \\

\multicolumn{1}{c}{} & \multicolumn{1}{c}{\textbf{NURS}} &
                                                            \multicolumn{1}{c}{\textbf{BFS}} &
                                                                                               \multicolumn{1}{c}{\textbf{DFS}} & \multicolumn{1}{c}{\textbf{allpath}} \\ 
\hline
\textbf{Virt} & $0/15$ & $0/15$ & $0/15$ & $0/15$ \\
\textbf{Virt $\times2$} & $0/15$ & $0/15$ & $0/15$ & $0/15$ \\
\textbf{Virt $\times3$} &  $1/15$ &  $1/15$ &  $1/15$ &  $2/15$ \\
\textbf{Flat-Virt} & $0/15$ & $0/15$ & $0/15$ & $0/15$ \\
\textbf{Flat-MBA} & $0/15$ & $0/15$ & $0/15$ & $0/15$ \\
     \hline
\textbf{\obsplit ($\times10$)} & $0/15$ & $0/15$ & $0/15$ & $0/15$ \\
\textbf{\obsplit ($\times13$)} & $1/15$ & $1/15$ & $0/15$ & $1/15$ \\
\textbf{\obfor ($k=1$)} & $0/15$ & $0/15$ & $0/15$ & $0/15$ \\
\textbf{\obfor ($k=2$)} &  $1/15$ &  $1/15$ &  $1/15$ &  $4/15$ \\
\textbf{\obfor ($k=3$)} &  $10/15$ &  $8/15$ &  $8/15$ &  $13/15$ \\
\textbf{\obfor ($k=4$)} &  $\showmore 15/15$ &   $\showmore 15/15$  
& $\showmore 15/15$ &  $\showmore 15/15$ \\

   \end{tabular}}
   \label{tab:secret-finding3}
\end{table}

\paragraph{Runtime overhead} We evaluate how the performance penalty evolved 
for protection \obfor  on very high values of $k$. We take the 15 examples of Dataset \#1 
with large input space, and we vary the size of the input string from 3 to
 100000, increasing the number of forking points accordingly ($k$ between 3 and 100000), 
one forking point (loop) per byte of the
 input string. We run 15 random inputs 15 times for each size and measure the
 average runtime overhead.  \cref{fig:execution-overhead} shows the evolution of
 runtime overhead w.r.t. the number of \obfor loops. 

\begin{figure}[!htbp]
  \centering
  \includegraphics[scale=0.5]{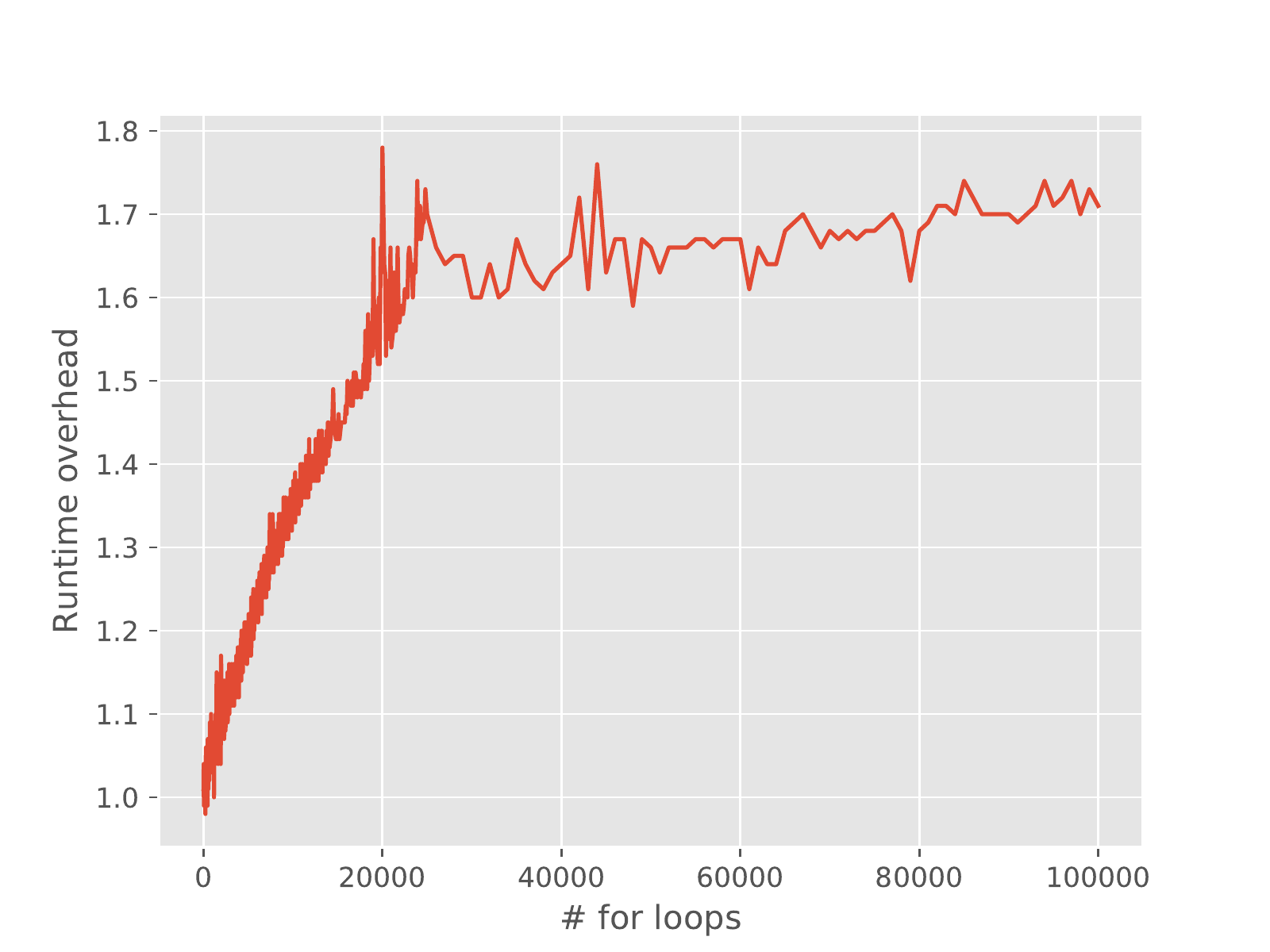}
  \caption{Runtime overhead w.r.t. to the number of \obfor obfuscation loops}
  \label{fig:execution-overhead}
\end{figure}

The runtime overhead stays below 5\% ($\times 1.05$) for
fewer than $k=250$. 
This means in particular that one can
significantly boost \obfor-based protections without incurring big
runtime penalties.

\begin{table*}[htbp]
    \caption{Benchmarking obfuscations on the crackme example -- tool \klee -- 1h30 timeout.}
    \label{tab:motiv-full}
    \centering
    \begin{tabular}{|c|c|l||R{1.5cm}|R{1.5cm}|R{1.5cm}|R{1.5cm}|}
        \cline{2-7}
        \multicolumn{1}{c|}{}
        & \multicolumn{2}{c||}{\multirow{2}{*}{\textbf{Obfuscation type}}}
        & \multicolumn{2}{C{3cm}|}{\textbf{Slowdown} \newline \itshape Symbolic Execution} 
        & \multicolumn{1}{c|}{\multirow{2}{*}{
            \shortstack{\textbf{Overhead }\\ \itshape Runtime}}} 
        & \multicolumn{1}{c|}{\multirow{2}{*}{
            \shortstack{\textbf{Overhead }\\ \itshape Code Size}}} \\ 
        \cline{4-5}
        \multicolumn{1}{c|}{} & \multicolumn{2}{c||}{} & Coverage & Secret &  &  \\
        \hline
        \multirow{7}{*}{\bf Tigress} & \multicolumn{2}{l||}{Virt $\times2$} & $\times5.4\cdot10^3$ & $\times11$ & $\times1.3$ & $\times1.2$ \\
        & \multicolumn{2}{l||}{ Virt $ \times3$} & {\showmore TO} & $\times1.1\cdot10^3$ & $\times41$ & $\times3.0$ \\
        & \multicolumn{2}{l||}{ Virt $ \times4$} & {\showmore TO} & $\times96\cdot10^3$ & $\times4.5\cdot10^3$ & $\times4.0$ \\
        & \multicolumn{2}{l||}{  Virt $  \times5$} &  {\showmore TO} &
                                                                       {\showmore
                                                                       TO} &  $\times449\cdot10^3$ &  $\times5.2$ \\
        & \multicolumn{2}{l||}{Virt-Flat} &  $\times1.0$ & $\times1.8$ & $\times1.1$ & $\times1.5$ \\
        & \multicolumn{2}{l||}{Flat $\times2$} & $\times276$ & $\times1.8$ & $\times1.1$ & $\times1.3$\\
        & \multicolumn{2}{l||}{Flat-EncA} & $\times83$ & $\times1.2$ &  $\showmore \times1.0$ &  $\showmore \times1.0$\\
        \hline\hline
        &&&&&&\\[-1em]
        \multirow{6}{*}{\textbf{Our approach}}&  & $k=1$ & {\showmore TO} & $\times3.3\cdot10^3$ &  $\showmore \times1.0$ &  $\showmore \times1.0$ \\
        & \obfor  &  $k=2$ &  {\showmore TO} &  {\showmore TO} &   $\showmore \times1.0$ &   $\showmore \times1.0$ \\
        & &  $k=3$ &  {\showmore TO} &  {\showmore TO} &   $\showmore \times1.0$ &   $\showmore \times1.0$ \\
        \cline{2-7}
        & & $k=11$ & $\times3.4\cdot10^3$ & $\times2.8$ &  $\showmore \times1.0$ &  $\showmore \times1.0$ \\
        & \obsplit & $k=15$ & TO & $\times3.9$ &  $\showmore \times1.0$ &  $\showmore \times1.0$ \\
        & & $ k=19$ &  {\showmore TO} &  $\times5.1$ &   $\showmore \times1.0$ &   $\showmore \times1.0$ \\
        \hline
    \end{tabular}
\end{table*}

\section{More details on experiments}

We give here more detailed  results on: 

\begin{itemize}

\item the motivating example  (\cref{tab:motiv-full}), 

\item Path Exploration Dataset \#1 (\cref{tab:path-exp}), 

\item Secret Finding DataSet \#1 (\cref{tab:secret-finding}). 

\end{itemize}

\begin{table*}[!htbp]
 \caption{Obfuscations on Dataset \#1 --- allpath coverage --- 3h timeout}
     \centering

     \resizebox{\textwidth}{!}{%
   \begin{tabular}{|l|r|r|r||r|r|r||r|r|r||c|}
\cline{2-11}
     \multicolumn{1}{c|}{}        & \multicolumn{3}{c||}{\textbf{DSE Slowdown}} & \multicolumn{3}{c||}{\textbf{Runtime overhead}} & \multicolumn{3}{c||}{\textbf{Code Size increase}} & \multirow{2}{*}{\textbf{\#TO}} \\
\cline{1-10}
\textbf{Transformation} & min & max & avg & min & max & avg & min & max & avg &\\
\hline\hline
\textbf{Virt} & $\times1.0$ & $\times17$ & $\times2.8$ & $\times1.2$ & $\times5.6$ & $\times1.5$ & $\times1.5$ & $\times1.5$ & $\times1.5$ & $0/46$ \\
\textbf{Virt $\times2$} & $\times1.0$ & $\times402$ & $\times47$ & $\times1.3$ & $\times432$ & $\times15$ & $\times2.3$ & $\times2.8$ & $\times2.5$ & $1/46$ \\
\textbf{Virt $\times3$} &  $\showmore \times10$ &  $\times35\cdot10^3$ &  $\times3.0\cdot10^3$ &  $\times3.2$ &  $\times52\cdot10^3$ &  $\times1.6\cdot10^3$ &  $\times3.5$ &  $\times4.6$ &  $\times4$ &  $5/46$\\
\textbf{Flattening} & $\times1.0$ & $\times1.3$ & $\times1.0$ &  $\showmore \times1.0$ & $\times2.0$ & $\times1.8$ &  $\showmore \times1.0$ &  $\showmore \times1.0$ &  $\showmore \times1.0$ & $0/46$ \\
\textbf{EncodeArithmetic} & $\times1.0$ & $\times10$ & $\times3.9$ & $\times1.0$ & $\times2.0$ & $\times1.8$ &  $\showmore \times1.0$ &  $\showmore \times1.0$ &  $\showmore \times1.0$ & $0/46$ \\
\hline\hline
&&&&&&&&&&\\[-1em]
\textbf{\obsplit ($k=10$)} & $\showmore\times10$ & $\times1.2\cdot10^3$ & $\times107$ &  $\showmore \times1.0$ & $\showmore\times1.3$ & $\times1.2$ &  $\showmore \times1.0$ &  $\showmore \times1.0$ &  $\showmore \times1.0$ & $1/46$ \\
\textbf{\obsplit ($k=13$)} & $\showmore\times10$ & $\times15\cdot10^3$ & $\times862$ &  $\showmore \times1.0$ & $\showmore\times1.3$ & $\times1.2$ &  $\showmore \times1.0$ &  $\showmore \times1.0$ &  $\showmore \times1.0$ & $4/46$ \\
\textbf{\obfor ($k=1$)} & $\showmore\times10$ & $\times476$ & $\times209$ &  $\showmore \times1.0$ & $\times1.4$ & $\times1.2$ &  $\showmore \times1.0$ &  $\showmore \times1.0$ &  $\showmore \times1.0$ & $2/46$ \\
\textbf{\obfor ($k=2$)} & $\showmore\times10$ & $\times33\cdot10^3$ & $\times3.7\cdot10^3$ &  $\showmore \times1.0$ & $\times1.4$ &  $\showmore \times1.0$ &  $\showmore \times1.0$ &  $\showmore \times1.0$ &  $\showmore \times1.0$ & $10/46$ \\
\textbf{\obfor ($k=3$)} & $\showmore\times10$ & $\times1.1\cdot10^6$ & $\times2.2\cdot10^5$ &  $\showmore \times1.0$ & $\times1.4$ & $\times1.3$ &  $\showmore \times1.0$ &  $\showmore \times1.0$ &  $\showmore \times1.0$ & $30/46$ \\
\textbf{\obfor ($k=4$)} & $\showmore\times10$ & $\showmore
                                                \times2.2\cdot10^6$ &
                                                                      $\showmore\times5.1\cdot10^5$ &  $\showmore \times1.0$ & $\times1.4$ & $\times1.3$ &  $\showmore \times1.0$ &  $\showmore \times1.0$ &  $\showmore \times1.0$ &  $\showmore 46/46$ \\
\hline\hline
\textbf{Virt $+$ \obfor ($k=2$)} &  $\showmore\times10$ &  $\times5.4\cdot10^5$ &  $\times33\cdot10^3$ &   $\showmore \times1.0$ &  $\times3.8$ &  $\times1.2$ &  $\times1.5$ &  $\times1.6$ &  $\times1.6$ &  $23/46$ \\
\hline
   \end{tabular}}
   \label{tab:path-exp}
\end{table*}

 \begin{table*}[!htbp]
   \caption{Obfuscations on Dataset \#1 --- ---
     Secret Finding --- DFS heuristics, 1h timeout}

   \centering

   \resizebox{\textwidth}{!}{%
   \begin{tabular}{|l|r|r|r|r|r|r|C{1cm}|}
     \cline{2-8}
     \multicolumn{1}{c|}{}        
     & \multicolumn{3}{c|}{\textbf{DSE slowdown}} & \multicolumn{3}{c|}{\textbf{Runtime overhead}} & \multirow{2}{*}{\textbf{\#TO}} \\
     \cline{1-7}
     \textbf{Transformation} & min & max & avg & min & max & avg & \\
     \hline\hline
     \textbf{Virt} & $\times1.0$ & $\times4.0$ & $\times1.6$ & $\times1.2$ & $\times1.4$ & $\times1.3$ & $0/15$ \\
     \textbf{Virt $\times2$} & $\times6$ & $\times268$ & $\times33$ & $\times1.3$ & $\times6.3$ & $\times2.5$ & $0/15$ \\
     \textbf{Virt $\times3$} &  $\times557$ &  $\times4.7\cdot10^3$ &  $\times1.7\cdot10^3$ &  $\times5.5$ &  $\times513$ &  $\times126$ &  $2/15$\\
     \textbf{Flat-Virt} & $\times1.0$ & $\times8.3$ & $\times2.3$ & $\times1.2$ & $\times1.5$ & $\times1.3$ & $0/15$ \\
     \textbf{Flat-MBA} & $\times2.0$ & $\times878$ & $\times59$ & $\times1.2$ & $\times1.3$ & $\times1.3$ & $0/15$ \\
     \hline\hline
     \textbf{\obsplit ($k=10$)} & $\times1.1$ & $\times9$ & $\times6$ &  $\showmore \times1.0$ & $\times1.3$ & $\times1.2$ & $0/15$ \\
     \textbf{\obsplit ($k=13$)} & $\times1.1$ & $\times12$ & $\times8$ & $\times1.3$ & $\times1.8$ & $\times1.6$ & $0/15$ \\
     \textbf{\obfor ($k=1$)} & $\times7$ & $\times1.1\cdot10^3$ & $\times169$ &  $\showmore \times1.0$ &  $\showmore \times 1.1$ &  $\showmore \times1.0$ & $0/15$ \\
     \textbf{For ($k=2$)} &  $\times841$ &  $\times1.7\cdot10^5$ &  $\times17\cdot10^3$ &   $\showmore \times1.0$ &   $\showmore \times 1.1$ &   $\showmore \times1.0$ &  $1/15$ \\
     \textbf{For ($k=3$)} &  $\times2.3\cdot10^3$ &  $\times3.6\cdot10^5$ &  $\times1.6\cdot10^5$ &   $\showmore \times1.0$ &   $\showmore \times 1.1$ &   $\showmore \times1.0$ &  $8/15$ \\
     \textbf{For ($k=4$)} &  $\showmore \times2.1\cdot10^5$ &  $\showmore
                                                 \times4.2\cdot10^5$ &
                                                                       $\showmore
                                                                       \times3.2\cdot10^5$
                      &   $\showmore \times1.0$ &   $\showmore \times
                                                  1.1$ &   $\showmore
                                                         \times1.0$ &
                                                                      $\showmore
                                                                      15/15$ \\
     \hline
 \end{tabular}}
   \label{tab:secret-finding}
 \end{table*}

\end{document}